\documentclass[a4paper,USenglish]{lipics-v2021}
\usepackage{main}

\title{Towards Communication-Efficient Peer-to-Peer Networks}



\author{Khalid Hourani}{Department of Computer Science, University of Houston, Houston, TX, USA}{mail@khourani.com}{https://orcid.org/0000-0002-2367-7124}{K. Hourani was supported in part by NSF grants CCF-1540512, IIS-1633720, and CCF-1717075 and BSF grant 2016419.}

\author{William K. {Moses Jr.}}{Department of Computer Science, Durham University, Durham, UK}{wkmjr3@gmail.com}{https://orcid.org/0000-0002-4533-7593}{Part of the work was done while William K. Moses Jr. was a postdoctoral fellow at the University of Houston, Houston, TX, USA. W. K. Moses Jr. was supported in part by NSF grants CCF-1540512, IIS-1633720, and CCF-1717075 and BSF grant 2016419.}

\author{Gopal Pandurangan}{Department of Computer Science, University of Houston, Houston, TX, USA}{gopal@cs.uh.edu}{https://orcid.org/0000-0001-5833-6592}{G. Pandurangan was supported in part by NSF grants CCF-1540512, IIS-1633720, and CCF-1717075 and BSF grant 2016419.}

\authorrunning{K. Hourani, W.\,K. Moses Jr., and G. Pandurangan}
\Copyright{Khalid Hourani, William K. Moses Jr., and Gopal Pandurangan}

\hideLIPIcs

\begin{document}

\maketitle

\begin{abstract}
    We focus on designing Peer-to-Peer (P2P) networks that enable  efficient communication.
Over the last two decades, there has been substantial algorithmic research on distributed protocols for building P2P
networks with various desirable properties such as high expansion, low diameter, and robustness to a large number of
deletions. A key underlying theme in all of these works is to distributively build a \emph{random graph} topology
that guarantees the above properties. Moreover, the random connectivity topology is widely deployed in many P2P systems
today, including those that implement blockchains and cryptocurrencies. However, a major drawback of using a random
graph topology for a P2P network is that the random topology does not respect the \emph{underlying} (Internet)
communication topology. This creates a large
\emph{propagation delay}, which is a major communication bottleneck in modern P2P networks.

In this paper, we work towards designing P2P networks that are communication-efficient (having small propagation delay)
with provable guarantees. Our main contribution is an efficient, decentralized protocol, $\WEAVER$, that transforms a random graph
topology embedded in an underlying Euclidean space  into a topology that also
respects the underlying metric.  We then present efficient  point-to-point routing and broadcast protocols that achieve
essentially optimal performance with respect to the underlying space. 

\end{abstract}

\thispagestyle{empty}

\setcounter{page}{0}
\newpage
\section{Introduction}
\label{sec:intro}
There has been a long line of algorithmic research on building Peer-to-Peer (P2P) networks (also called overlay networks)
with desirable properties such as connectivity, low diameter, high expansion, and robustness to adversarial deletions
\cite{PRU01,LS03,mihail-p2p,CDG07-soda,CDG07,JP12,focs2015}. 
A key underlying theme in  these works is a distributed protocol to build a \emph{random graph} that
guarantees {these
        desirable properties}. The high-level idea is for a node to connect to a small, but \emph{random}, subset of nodes. In
fact, this \emph{random connectivity} mechanism is used in real-world P2P  networks. For example, in the Bitcoin P2P
network, each node connects to 8 neighbors chosen in a random fashion~\cite{perigee}. It is well-known that a random
(bounded-degree) graph is an expander with high probability.\footnote{In this paper, by expander, we mean one with bounded degree, i.e., the
    degree of all nodes is bounded by a constant or a slow-growing function of $n$, say $\bigO{\polylog{n}}$, where $n$ is
    the network size.}\footnote{For a graph with $n$ nodes, we say that it has a property with high
    probability when the probability is at least $1 - 1/n^c$ for some $c \geq 1$.} An expander graph on $n$ nodes has
high expansion  and conductance, low diameter (logarithmic in the network size) and
robustness to adversarial deletions --- even deleting $\epsilon n$ nodes (for a sufficiently small constant $\epsilon$) leaves
a giant component of $\bigTh{n}$ size which is also an expander~\cite{Hoory,bagchi}.

{Unfortunately}, a major drawback of using a random graph as a P2P network is that the connections are made to random nodes and
do not respect the \emph{underlying} (Internet) communication topology. This causes a large \emph{propagation latency}
or \emph{delay}. Indeed, this is a crucial problem in the Bitcoin P2P network, which has delays as high as 79 seconds on average~\cite{scalingblockchains,perigee}. A main cause for the delay is that the P2P (overlay) network induced by random
connectivity can be highly sub-optimal, since it {ignores} the underlying Internet communication topology (which depends
on geographical distance, among other factors).\footnote{It also ignores differences in bandwidth, hash-{strength},
and computational {power} across peers as well as malicious peers. Addressing these issues is beyond the scope of this paper.} 
The main problem we address in this paper is to show how one can efficiently modify a given random graph topology to
build P2P networks that also have {small \emph{propagation latency}}, in addition
to other properties such as low (hop) diameter and high expansion, with provable guarantees.

\onlyLong{
    Towards this goal, following prior work (see e.g.~\cite{perigee,vivaldi}),
    we model a P2P network as a random graph embedded in an underlying Euclidean space. This model is a reasonable approximation
    to a random connectivity topology on nodes distributed on the Internet (details in Section~\ref{subsec:model}).

    The main contributions of this paper are:
    \begin{enumerate*}[label=(\arabic*)]
        \item a theoretical framework to rigorously quantify performance of P2P communication protocols;
        \item $\WEAVER$, an efficient decentralized  protocol  that converts the random graph topology into a topology that
              also respects the underlying embedding;
        \item efficient point-to-point routing and broadcast algorithms in the modified topology that achieve
                  {essentially} optimal performance. We note that broadcasting is a key application used in
              P2P networks that implement blockchain and cryptocurrencies in which a block must be quickly
              broadcast to all (or most) nodes in the network.
    \end{enumerate*}
}

\subsection{Motivation, Model, and Definitions}
\label{subsec:model}

\textbf{Motivation.}
We consider a random graph network that is used in several prior P2P network construction protocols
(e.g.~\cite{PRU01,LS03,CDG07-soda,focs2015}). As mentioned earlier, real-world P2P networks, such as Bitcoin, also seek
to achieve a random graph topology (which are expanders with high probability~\cite{MUbook,Hoory}). 
Indeed, random graphs have
been used extensively to model P2P networks (see e.g.
\cite{LS03, PRU01, mihail-p2p, CDG07-soda, Mahlmann_2006}).

Before we formally state the model that is based on prior works~\cite{perigee,vivaldi,zihu}, we explain the motivation
behind it; we refer to~\cite{perigee} for more details and give a brief discussion here. Many of today's P2P (overlay)
networks employ the random connectivity algorithm; in fact, this is widely deployed in many cryptocurrency systems.\footnote{In particular, the real-world Bitcoin P2P network, constructed by allowing each node to choose 8 random
    (outgoing) connections ({\cite{Miller2015DiscoveringB,perigee}}) is likely an expander network if the
    connections are chosen (reasonably) uniformly at random~\cite{Palmer_1985_Book}.}
In this algorithm, nodes maintain a small number of connections to other nodes chosen in a random fashion. In
such a topology, for any two nodes $u$ and $v$, any path (including the shortest path) would likely go through nodes
that are not located close to the shortest geographical route (i.e., the geodesic) connecting $u$ and $v$. Such paths
that do not respect the underlying geographical placement of nodes often lead to higher propagation delay. Indeed, it
can be shown that a random topology yields paths with propagation delays much higher than those of paths on topologies
that respect the underlying geography~\cite{perigee}.

To model the underlying propagation costs, several prior works (see e.g., the Vivaldi system~\cite{vivaldi}) have empirically shown that nodes on the Internet can be embedded on a low-dimensional metric
space (e.g., $\reals^5$) such that the  distance between any two nodes accurately captures the communication
delay between them. In fact, the Vivaldi system demonstrates that even embedding the nodes in a 2-dimensional metric
space (e.g., $\mathbb{R}^2$) and using the corresponding distances captures the communication delay quite
well.  In contrast, the paths on a random graph topology
are highly sub-optimal, since they are unlikely to follow the optimal path on the embedded metric space.

The work of Mao et al.~\cite{perigee} illustrates the above disparity using the following example motivated by the above
discussion. Consider a network embedded in the unit square $[0,1] \times [0,1]$\shortOnly{.} \longOnly{(see Figure~\ref{fig:rg}).
    \begin{figure}[H]
        \centering
        \begin{tikzpicture}[scale=5]
    \draw (0, 0) rectangle (1, 1);
    \node[point,label=below:{$u$}] at (0.1505883063288043,0.2566482354253204)(n0){};
    \node[point] at (0.9602742655353219,0.3188449339305045)(n1){};
    \node[point] at (0.8697638197573846,0.14414257408706443)(n2){};
    \node[point] at (0.9345415490626012,0.49788489082499854)(n3){};
    \node[point] at (0.053785343752909776,0.5178746556929065)(n4){};
    \node[point] at (0.4964854474921945,0.8793193449836291)(n5){};
    \node[point] at (0.615405692376732,0.2394921634168461)(n6){};
    \node[point] at (0.1358275871399648,0.34926941874308637)(n7){};
    \node[point] at (0.8684990361615856,0.570272858250284)(n8){};
    \node[point] at (0.7222944016100679,0.6653447639165281)(n9){};
    \node[point] at (0.6941886548288694,0.8707138495965842)(n10){};
    \node[point] at (0.109380742084303617,0.03749300898355845)(n11){};
    \node[point] at (0.10823569561474733,0.5492816855154502)(n12){};
    \node[point] at (0.8187693436456968,0.7522472690559413)(n13){};
    \node[point] at (0.19171454640814014,0.054643425516101307)(n14){};
    \node[point] at (0.30573384324520414,0.36054530449885724)(n15){};
    \node[point] at (0.7878458521961231,0.09702413383994457)(n16){};
    \node[point] at (0.8434101860790817,0.43902888150510866)(n17){};
    \node[point] at (0.7817328213070882,0.09092764046719048)(n18){};
    \node[point,label=above:{$v$}] at (0.805069527577571287,0.8722146626203203)(n19){};
    \draw (n0) -- (n1) -- (n3) -- (n7) -- (n13) -- (n2) -- (n16) -- (n12) -- (n19);
    \draw[thick, dashed, blue] (n0) -- (n19);
\end{tikzpicture}
        \caption{Nodes in a random graph topology embedded in the unit square. Notice that the straight-line distance from $u$ to $v$ is much shorter
            than sum of the distances in a shortest path between the two nodes.}
        \label{fig:rg}
    \end{figure}
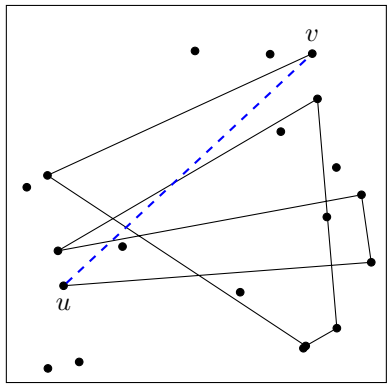
}
The set of nodes (points) $V$ {is drawn uniformly at random within the square}. The Euclidean distance, $\norm{u-v}_2$,
between any two nodes $u, v \in V$ represents the \emph{delay} or \emph{latency} of sending a message from $u$ to $v$
(or vice-versa). We construct a \emph{random graph} on $V$ by connecting each node in the unit-square randomly to a small
constant {number of} other nodes. \longOnly{Figure~\ref{fig:rg} shows the shortest path on this topology between two nodes,
    say $u$ (bottom left corner) to $v$ (top right corner).} Since the random graph does not respect the underlying geometry,
the propagation cost  between $u$ and $v$ --- defined as the sum of the Euclidean distances
of the edges in the shortest path --- is significantly greater than the point-to-point (geodesic) distance ($\norm{u - v}_2$) between them. We formally show this in Theorem~\ref{thm:gstretch}.

By comparison, consider a \emph{random geometric} graph on $V$, where the uniformly distributed
nodes are connected as follows: any two nodes $u$ and $v$ are connected by an edge if they are within a distance $\rho$ of
each other~\cite{Penrose,bincovering-journal}.\onlyLong{\footnote{Equivalently, one can connect each node to its $k$ closest nodes;
        one can show that, for appropriate values of $k$ and $r$, these two models have very similar properties
        \cite{bincovering-journal}.}} This model, called the $G(n,\rho)$ \emph{random geometric graph} model, has connections which respect the underlying geometry. In this graph, we can show that the shortest path between any two nodes $u$
and $v$ is much closer to the geodesic shortest path (the straight line path) between $u$ and $v$~\cite{bincovering-journal,friedrich}.

We note that, while the work of Mao et al.~\cite{perigee} showcases the disparity between a random graph and a random
geometric graph {(as discussed above)} it does not give any theoretical results on how to convert a random graph topology
into a random geometric graph topology.
On the other hand, it gives heuristics to transform
a P2P graph constructed on real-world data to a graph that has smaller propagation delays. The heuristics
are based on rewiring edges to favor edges between nodes that have smaller round-trip delays. It presents
experimental simulations to show that these heuristics do well in practice.  However, they do not formally analyze their algorithm and do not give any theoretical guarantees.

~\\\textbf{Network Model.}
~Motivated by the above discussion, and following prior works~\cite{perigee,vivaldi,zihu}, we model a P2P network $G$ as follows. We
assume $G$ to be a $d$-regular expander (where $d$ is a constant).\footnote{We assume a $d$-regular graph for
    convenience; we could have also assumed that the degree is bounded by some small growing function of $n$, say
    $\bigO{\log n}$.} Note that our results will also apply if $G$ has a random connectivity topology modeled by a
$d$-regular (or bounded degree) \emph{random graph} or a $G(n,p)$ random graph (with $p = \bigTh{\sfrac{\log n}{n}}$).
We note that such random graphs are expanders with high probability (see Definition~\ref{def:expander-graph}) \cite{KS11}. Our
model is quite general in the sense that we only assume that the topology is an expander; no other special properties
are assumed. \onlyLong{(Indeed, expanders have been used extensively to model P2P networks~\cite{LS03, PRU01, mihail-p2p, CDG07-soda, Mahlmann_2006, Augustine_2012, focs2015}.)}
Furthermore, we assume that the nodes of $G$ correspond to points that are distributed \emph{uniformly at random} in a
unit square $[0,1]\times[0,1]$.\footnote{Our model can be generalized to higher dimensions by embedding nodes in an $m$-dimensional hypercube $[0,1]^m$.}
Although, the assumption of nodes being uniformly distributed is strong, based on our experiments on the Bitcoin P2P network, this  appears to be a reasonable first approximation.\footnote{We embedded nodes in the Bitcoin P2P network in a 2-dimensional grid using the Vivaldi algorithm and although there were many outliers, a significant subset of nodes ended up being reasonably uniformly distributed.} Considering more general distribution models is a good direction for future work (cf. Section \ref{sec:related}).

We assume each node $u$ knows its ID and, while node $u$ \emph{need not} know its coordinates, it is able to determine its distance (which captures propagation delay)
to any node $v$ given only the ID of $v$.\footnote{Note that this assumption has to do only with implementing our protocols
    in a localized manner (which is relevant in practice) and does not affect their correctness or efficiency. In the Internet, for example, point-to-point propagation delay can be measured locally: a node can  determine the round-trip-time to another node
    using the \texttt{ping} network utility~\cite{pingbsd}.  On the other hand, it is also possible for a node
    to determine its coordinates --- as mentioned earlier, systems such as Vivaldi~\cite{vivaldi} can
    assign coordinates in a low dimensional space (even $\mathbb{R}^2$) that accurately
    capture the propagation delay between nodes.} In particular, we assume for convenience that
    a node can determine the Euclidean and the Manhattan distances (i.e., $L_2$ and $L_{\infty}$ norms respectively) between itself
    and another node if it knows the ID of that node.

An important assumption is that nodes initially have only \emph{local} knowledge, i.e., they have knowledge of only
themselves and their neighbors in $G$. In particular, they do \emph{not} have any knowledge of the global topology or of
the IDs of other nodes (except  those of their neighbors) in the network. We assume that nodes have knowledge of
the network size $n$ (or a good estimate of it).

We assume a synchronous network where computation and communication proceeds in a sequence of discrete \emph{rounds}.
Communication is via message passing on the edges of $G$. Note that $G$ is a P2P (overlay) network in the sense that a
node $u$ can communicate (directly) with another node $v$ if $u$ knows the ID of $v$. This is a typical assumption in
the context of P2P and overlay networks, where a node can establish communication with another node if it knows the
other node's IP address, and has been used in several prior works (see e.g.~\cite{focs2015,Augustine_2013_PODC,JP12,Pandurangan_2016,time-overlay}).
Note that $u$ can know the ID of $v$ either directly, because $u$ and $v$ are neighbours in $G$, or indirectly, through
received messages. In the latter case, this is equivalent to adding a \emph{``virtual''} edge between $u$ and $v$. Since
we desire efficient protocols, we require each node to send and receive messages of size at most $\polylog{(n)}$ bits in
a round. In fact, a node will also communicate with only $\polylog{(n)}$ other nodes in a round. Additionally,
the number of bits sent per edge per round is $\bigO{\polylog{(n)}}$.

\subsection{Preliminaries}
\label{sec:defns}
We need the following concepts before we formally state the problem that we address and our contributions.\\

\noindent \textbf{Embedded Graph.} We  define an \textit{embedded graph} as follows.
\begin{definition}\label{def:embedded}
    Let $G = (V, E)$ be any graph and consider a random embedding of the nodes $V$ into the unit square, i.e., a uniform and independent
    assignment of coordinates in $[0, 1] \times [0, 1]$ to each node in $V$. This graph, together with this embedding, is called an
    \emph{embedded graph}, and we induce weights on the edge set $E$, with the weight of an edge $(u, v)$ equal to the Euclidean
    distance between the coordinates assigned to $u$ and $v$, respectively.
\end{definition}

\noindent \textbf{Routing Cost.} We next define \emph{propagation cost} to capture the cost of routing along a path in an embedded graph.
\begin{definition}\label{def:prop-cost-def}
    Let $G$ be an embedded  graph. For any path $P = (v_1, v_2, \dots, v_{k - 1}, v_k)$
    the \emph{propagation cost}, also called the routing propagation cost, of the path $P$ is the weight of the path $P$ given by
    $d_G(P) = \sum_{i=1}^{k - 1}d(v_i, v_{i + 1})$,
    i.e., the sum of the weights (the Euclidean distances) of edges along the path. The
    value $k$, denoted by $\hopcount_G(P)$, is the \emph{hop count} (or hop length) of the path $P$.
    The \emph{minimum propagation cost} between the nodes $u$ and $v$  is the weight of the shortest path
    between $u$ and $v$ in the embedded graph.
\end{definition}

Note that the propagation cost between two nodes is lower bounded by the Euclidean distance between them.
Given two nodes, we would like  to route  using a path of \emph{small} propagation cost,
i.e., a path whose propagation cost is  close  to the Euclidean distance between the two nodes. In particular,
we would like the ratio between the two to be small. (We would also like the hop count to be small.)

The following theorem shows that, in a \emph{$d$-regular random graph} $G$ embedded in the unit square, the \emph{ratio}
of the propagation cost of the shortest path between two nodes $u$ and $v$ in $G$ to the Euclidean distance between those nodes  can be as high as $\Omega(\sqrt{n})$
on expectation.
Thus a P2P topology that is modeled by a random graph topology has a high propagation cost for some node pairs.

\onlyLong{
    We first prove the following lemma.

    \begin{lemma*}\label{lem:embedded-nodes}
        Let $G$ be an embedded (connected) $d$-regular graph, for any constant $d \geq 2$. There exist a pair of nodes,
        $u$ and $v$, such that $d(u, v) = \bigO{\sfrac{1}{\sqrt{n}}}$ and all shortest paths $P$ between $u$ and $v$ have $\bigOm{\log{n}}$ hops and $\expectation{d(P)}$,
        the expected value of the sum of the Euclidean distances of the edges along this path, is $\bigOm{1}$.
    \end{lemma*}

    \begin{proof}
        Fix a node $u$. Since $G$ is $d$-regular, we can show that there exists  a set $S$ of $\epsilon n$ nodes (for some fixed constant $\epsilon > 0 $ that depends on $d$) that is at least $c\log{n}$ hops away from $u$ (for a suitably small constant $c > 0$). Fix this set $S$ of
        $\epsilon n$  nodes.
        Consider the square of side-length $\sfrac{1}{\sqrt{n}}$ (and thus
        area $1/n$) centered
        at $u$ --- with probability
        \[\binom{\epsilon n}{1}\frac{1}{n}\left(1 - \frac{1}{n}\right)^{\epsilon n - 1} = \epsilon\left(1 - \frac{1}{n}\right)^{\epsilon n - 1} \approx \epsilon\exp\left(-\epsilon\frac{n-1}{n}\right) = \bigTh{1} \]
        exactly one of nodes in the set $S$  (call it $v$) falls within this square. Clearly,
        $d(u, v) = \bigO{\sfrac{1}{\sqrt{n}}}$.

        Now, consider any path from $u$ to $v$; any such path has to go through a neighbor of $u$. The probability
        that a single node has distance at least $\delta$ from $u$ is simply the area of the circle of radius
        $\delta$ centered at $u$, which is $\pi\delta^2$. Hence for a fixed constant $\delta > 0$,  with  probability
        $\left(1 - \pi\delta^2\right)^d = \Theta(1)$
        all $d$ neighbors of $u$ in $G$ (note that these $d$ neighbors are disjoint from the nodes in set $S$ that are at least $c\log{n}$ hops away) will have distance at least $\delta$ from $u$. Thus, for constant $\delta$,
        at least one edge on every path from $u$ to $v$ will have expected length $\bigOm{1}$ and, by
        the triangle inequality, every such path has distance at least $\bigOm{1}$.
    \end{proof}
}

\begin{theorem}
    \label{thm:gstretch}
    Let $G$ be a $d$-regular random  graph embedded in the unit square.
    Then, there exists a pair of nodes $u$ and $v$ in $G$ such that $\sfrac{d_G(P)}{d(u, v)}$, the ratio
    of the propagation cost of the shortest path $P$ between $u$ and $v$ to the Euclidean distance between them
    is $\bigOm{\sqrt{n}}$ on expectation.
\end{theorem}

\onlyLong{
    \begin{proof}
        Take $u$ and $v$ as in Lemma~\ref{lem:embedded-nodes}. Note that $d(u, v) = \bigO{\sfrac{1}{\sqrt{n}}}$ and
        that any shortest path between $u$ and $v$ has expected total distance $\bigOm{1}$. Thus, the ratio of the propagation cost of the shortest path between them to the Euclidean distance between them is simply
        \[\frac{\bigOm{1}}{\bigO{\sfrac{1}{\sqrt{n}}}} = \bigOm{\sqrt{n}}.\qedhere\] 
    \end{proof}
}

We use propagation cost to measure the performance of a routing algorithm in $G$. The goal is to construct a graph topology so that one can find  paths of small propagation costs between every pair of nodes.  Moreover, we want a routing algorithm
that routes along paths of small propagation cost while also keeping the hop length small. \\

\noindent {\bf Broadcast Performance Measures.}
Next, we quantify the performance of a broadcasting algorithm.

\begin{definition}
    Consider a broadcast algorithm $\mathcal{A}$ that broadcasts a single message from a given source to
    all other nodes in some connected embedded graph $G$. The \emph{broadcast propagation cost} of algorithm
    $\mathcal{A}$ on graph $G$ is defined as the the sum of the Euclidean distance of the edges used by
    $\mathcal{A}$ to broadcast the message.
\end{definition}

Notice that the broadcast propagation cost roughly captures the efficiency of a broadcast algorithm.
We note that the best possible broadcast propagation cost for a graph is
broadcasting by using \emph{only} the edges of the minimum spanning tree (MST) on $G$.
In particular, this yields the following lower bound for a graph whose nodes are embedded uniformly at random
in the unit square.  The proof follows from a bound on the weight of a Euclidean MST on a set
of points distributed uniformly in a unit square \cite{steele}.

\begin{theorem}\label{the:lower-bound-broadcast}[follows from \cite{steele}]
    The broadcast propagation cost of any algorithm $\mathcal{A}$ on an embedded graph $G$ whose nodes
    are distributed uniformly in  a unit square is
    $\bigOm{\sqrt{n}}$ with high probability.
\end{theorem}

On the other hand, we show  that the broadcast propagation cost of the standard flooding algorithm \cite{peleg-book} on a random graph embedded in a unit square
is high compared to the above lower bound.

\begin{theorem}
    \label{thm:gbroadcast}
    Let $G$ be a $d$-regular random graph embedded in the unit square.
    The standard flooding algorithm on $G$ has $\bigTh{n}$ \emph{expected}
    broadcast propagation cost.
\end{theorem}

\onlyLong{
    \begin{proof}
        The message will be sent across every edge at least once. Thus,
        $\sum_{(u, v) \in E}w(u, v)$
        is a lower bound on the propagation cost, where $w(u,v)$ is
        the Euclidean distance between two nodes $u$ and $v$. 
        We use the principle of deferred decisions to bound the expected value of the weight (Euclidean distance) of an edge.
        Fix an edge $e=(u,v)$ in the graph and consider its expected length.
        Since $u$ and $v$ are chosen uniformly at random in the unit square, it is easy to show that the $\expectation{w(u,v)} = \Theta(1)$.
        Hence  by linearity of expectation
        \[\sum_{(u, v) \in E}w(u, v) = \frac{dn}{2} \cdot \bigTh{1} = \bigTh{n}.\qedhere\]
    \end{proof}
}

We also use  other metrics to measure the quality of a broadcast algorithm $\mathcal{A}$.
The \emph{broadcast completion cost} and \emph{broadcast completion time} measure, respectively, the propagation cost
and the number of hops needed to reach any other node $v$ from a given source $u$.

\begin{definition}
    Consider a broadcast algorithm $\mathcal{A}$ that broadcasts a single message from some source node $z$ to all
    other nodes in some connected graph $G(V,E)$. The \emph{broadcast completion cost of $\mathcal{A}$ on $G$} is the
    maximum value of the minimum propagation cost between the source node $s$ and any node $u$ considering paths taken
    by the message in $\mathcal{A}$, taken over all nodes $u \in V$ and all possible source nodes $s \in V$. 
    More precisely, let $\Prop_{\mathcal{A}}(s,u)$ be the minimum propagation cost for a message sent from the node $s$
    to reach node $u$ using broadcast algorithm $\mathcal{A}$ {and define} $\Prop_{\mathcal{A}}(s) = \max_{u\in V} \Prop_{\mathcal{A}}(s,u)$.
    Then, broadcast completion cost is $\max_{s \in V} \Prop_{\mathcal{A}}(s)$. The \emph{broadcast completion time} of
    $\mathcal{A}$ on $G$ is simply the number of rounds before the message from the source node reaches all nodes.
\end{definition}

\noindent \textbf{Conductance and Expanders.}
We recall the notions of conductance of a graph and that of an expander graph.
\begin{definition}[Conductance]
    The conductance $\phi(G)$ of a graph $G= (V,E)$ is defined as:
    $\phi(G) = \min_{S \subseteq V} \frac{\card{E(S,\comp{S})}}{\min\set{\vol{S}, \vol\comp{S}}}$
    where, for any set $S$, $E(S,\comp{S})$ denotes the set of all edges with one vertex in $S$ and one vertex in
    $\comp{S} = V-S$, and
        {$\vol(S)$, called the
            \emph{volume} of $S$, is the sum of the degrees of all nodes in $S$}.
\end{definition}

\begin{definition}[Expander Graph]\label{def:expander-graph}
    A family of graphs $G_n$ on $n$ nodes is an \emph{expander family} if,
    for some constant $\alpha$ with $0 < \alpha < 1$, the conductance
    $\phi_n = \phi(G_n)$ satisfies $\phi_n \geq \alpha$ for all
    $n \geq n_0$ for some $n_0 \in \mathbb{N}$.
\end{definition}

\noindent {\bf Random Geometric Graph.}
\begin{definition}[Random Geometric Graph]\label{def:rgg}
    A random geometric graph, $G(n,\rho)=(V,E)$, is a graph
    of $n$ points, independently and uniformly at random placed
    within $[0,1] \times [0, 1]$ (the unit square). These points
    form the node set $V$, and for two nodes $u$ and $v$, $(u,v)\in E$
    if and only if the distance $d(u,v)$ is at most $\rho$, for parameter
    $0 < \rho = f(n)  \leq 1$.
\end{definition}

We note that the distance between points is the standard Euclidean distance. The $G(n, \rho)$ graph exhibits the threshold phenomenon for many
properties, such as connectivity, coverage, presence of a giant component,
etc.~\cite{Penrose,bincovering-journal}. For example, the threshold for connectivity is $\rho = \bigTh{\sqrt{\sfrac{\log n}{n}}}$, i.e., if the value of $\rho$ is $\bigOm{\sqrt{\sfrac{\log n}{n}}}$, the graph $G(n,\rho)$ is connected with high probability;
on the other hand, if $\rho = \littleO{\sqrt{\sfrac{\log n}{n}}}$, then the graph
is likely to be disconnected. It is also known~\cite{G18} that the diameter of $G(n,\rho)$ (above the connectivity threshold) is $\tilde{\Theta}(1/\rho)$ with high probability.\footnote{Throughout, the $\tilde{O}$
    notation hides a $\operatorname{polylog} n$  factor and $\tilde{\Omega}$ hides a $1/(\operatorname{polylog} n)$ factor.}

\subsection{Problems Addressed and Our Contributions}
\label{sec:contributions}

As shown in Theorems~\ref{thm:gstretch} and~\ref{thm:gbroadcast}, routing (even via the shortest path) and the
standard flooding broadcast protocol in an embedded \emph{random} graph $G$ have a relatively large
point-to-point routing propagation cost and broadcast propagation cost, respectively.

Given a P2P network modeled as a random graph $G$ embedded on a unit square, the goal is to design an efficient distributed protocol to
transform  $G$ into a network $G^*$ that admits efficient communication primitives for the fundamental tasks of routing and broadcast, in particular, those that have essentially optimal routing and broadcast propagation costs.
Furthermore, we want
to design optimal routing and broadcast protocols on $G^*$. (Broadcasting is a key application used in
              P2P networks that implement blockchain and cryptocurrencies in which a block must be quickly
              broadcast to all (or most) nodes in the network.)

\bigskip

\noindent Our contributions are as follows:
\begin{enumerate}[label=\textbf{\arabic*.}]
    \item We develop a theoretical framework to model and analyze P2P network protocols, specifically point-to-point routing and broadcast (see Section~\ref{subsec:model}).

    \item We present an efficient distributed P2P topology construction protocol, $\WEAVER$, that takes a
          P2P expander network $G$  and improves it  into a  topology $G^*$
          that admits essentially optimal routing and broadcast primitives (see Section~\ref{sec:reconfig-alg}).
          Our protocol uses only local knowledge and is fast, using only $O(\polylog{n})$ rounds.  $\WEAVER$
          is based on random walks which makes it quite lightweight (small local computation overhead) and inherently decentralized and robust (no single point of action, no construction of tree structure, etc).
          It is also scalable in the sense
          that each node sends and receives only $O(\polylog{n})$ bits per round and communicates with only $O(\polylog{n})$
          nodes at any round. We assume only that the given topology $G$ is an expander graph; in particular,
          $G$ can be random graph (modeling a random connectivity topology, see Section~\ref{subsec:model}).

    \item To show the efficiency of $G^*$, we develop a distributed routing protocol $\GREEDYROUTING$ as well as broadcast protocols $\FLOODING$ and $\BROADCAST$
          that have essentially optimal routing and broadcast propagation costs, respectively (see Section \ref{sec:applications-graph}).
\end{enumerate}

\shortOnly{For lack of space, we refer to the full paper (in appendix) for proofs, additional details, and figures.}

\longOnly{
\subsection{Technical Overview}
\subsubsection{\texorpdfstring{$\WEAVER$}{Weaver} protocol}
~~~~The high-level idea behind  $\WEAVER$ is as follows. Starting from an expander graph, the goal is to construct
a topology that (i) contains a random geometric graph and (ii) contains a series of graphs such that each graph is an expander with edges that respect a maximum upper bound on Euclidean distance, for various distance values. This topology will then be used
to design efficient communication protocols (Section~\ref{sec:applications-graph}). To construct a random geometric graph, nodes must discover other nodes that
are close to them in the Euclidean space; in particular, each node $u$ needs to connect to all
nodes within distance $\rho = \Theta(\sqrt{\log n/n})$ to form a connected random geometric graph (see Definition~\ref{def:rgg}). The challenge is, in the given expander graph, nodes \emph{do not} have knowledge of the IDs of other nodes (except their neighbors in the original graph) in the network. The protocol allows each node to find nodes that are
progressively closer in distance to itself. From $G$, which is an expander in the  unit square, we construct several expander graphs in squares of smaller side-length. The expander graph is constructed by each node connecting to a small number of random nodes in the appropriate square. This creates a random graph which we show to be an expander with high probability (Lemma~\ref{lem:graph-in-each-phase-expander}).
Connecting to random nodes is accomplished by performing \emph{lazy random walks}  which mix fast (i.e., reach
the uniform stationary distribution) due to the expander graph
property \cite{Hoory}. We show a key technical lemma (Lemma \ref{lem:graph-in-each-phase-expander})
that proves the expansion property of the expanders created by random walks by analyzing the conductance
of the graph. By constructing expander graphs around each node  in progressively smaller areas, the protocol finally
is able to locate all nearby (within distance $\rho = \Theta(\sqrt{\log n/n})$) nodes with high probability. Then each node
forms connections to these nodes, which guarantees that a random geometric graph is included as a subgraph.
The final constructed graph $G^*$ includes all the edges that were created by the protocol, in addition
to the edges in the original graph $G$. Thus, $G^*$ is an expander (with degree bounded by $O(\log^2 n)$) and also includes
a random geometric graph as a subgraph (among other  edges added by nodes to random nodes at varying distances).

\subsubsection{P2P Routing Protocol}~~We note that the original graph $G$ does not admit an efficient point-to-point routing protocol as it is
a random graph and is not addressable.  Note
that even if one uses shortest path routing (assuming shortest paths have been constructed a priori),
the propagation cost can be as high as $\Omega(\sqrt{n})$ (see Theorem~\ref{thm:gstretch}).

We present a P2P routing protocol, called $\GREEDYROUTING$, with near-optimal
propagation cost in $G^*$. The key benefit of our routing protocol is that it is \emph{fully-localized}, i.e.,
a node $u$ needs only the ID of the destination node $v$ and the distances of its neighbors to $v$ to determine to which neighbor
it should forward the message. In particular, we show that a simple greedy protocol that always forwards the message to a neighbor closest
to the destination correctly and efficiently routes the message. Routing protocols that assume that each node knows its own position and that of its neighbors and that the position of the destination is known to the source are sometimes referred to as geometric routing and greedy approaches to such routing have been explored in the literature (e.g,~\cite{KSU99},~\cite{KK00},~\cite{KWZY03} and the references therein). We show that $\GREEDYROUTING$ takes $\bigO{\log n}$ hops to reach the destination and, more importantly,
that the propagation cost is close to the optimal propagation cost  needed to route between the two nodes.
Our protocol carefully exploits the geometric structure of the constructed edges created in $G^*$ to show
the desired guarantees.

\subsubsection{P2P Broadcast Protocols}~~We see that any broadcast algorithm that runs on the type of input graphs we consider takes at least $\Omega(\sqrt{n})$ broadcast propagation cost with high probability (Theorem~\ref{the:lower-bound-broadcast}). However, if we ran a simple flooding algorithm on the original graph, we could only achieve a broadcast propagation cost of $\Theta(n)$ on expectation (see Theorem~\ref{thm:gbroadcast}). In contrast, we develop a broadcast algorithm that leverages the structure of $G^*$ to achieve broadcast with broadcast propagation cost $O(\sqrt{n \log^3 n})$, which is asymptotically optimal up to polylog $n$ factors. The challenge is to carefully select edges in $G^*$ to send the message over.

If we simply flood the message over the edges of the random geometric graph contained within $G^*$, as we do in $\FLOODING$, we would obtain the desired broadcast propagation cost and an optimal (up to polylog $n$ factors) broadcast completion cost of $\tilde{O}(1)$, however the broadcast completion time, i.e., the number of rounds (hops) until the message reaches all nodes, is $\tilde{\Omega}(\sqrt{n})$, which is large. Thus, we turn to a more delicate algorithm, $\BROADCAST$, that results in the optimal (up to $\polylog{n}$ factors) broadcast propagation cost, optimal broadcast completion cost of $O(1)$, and an optimal (up to polylog $n$ factors) broadcast completion time of $O(\log n)$. 
However, in order to achieve such guarantees, we make use of the stronger assumption that nodes know their coordinates in the unit square, instead of merely knowing the distances between themselves and other nodes.

$\BROADCAST$ is described in more detail below. Consider a partition, call it $H_i$, of the unit grid into a $1/r^i$ by $1/r^i$ grid of $1/r^{2i}$ equal size squares where $0 < r < 1$ is a constant. By carefully sending the message to one node per square in $H_{\Theta(\log n - \log \log n)}$, and then performing simple flooding over the random geometric graph in $G^*$, we can achieve broadcast with the desired values for all three metrics. That is, the broadcast propagation cost is $O(\sqrt{n \log^3 n})$ and the broadcast completion time is $O(\log n)$, both of which are optimal up to polylog $n$ factors, and the broadcast completion cost is $O(1)$, which is an optimal bound. 
}

\subsection{Other Related Work}
\label{sec:related}

There are several works (see e.g., \cite{angluin,deconstructor,time-overlay})  that begin with an arbitrary graph and reconfigure it to be an expander (among other topologies). The expander topology constructed does not deal with the underlying (distance) metric. Our work, on the other hand, starts with an arbitrary expander topology (and here one can use algorithms such as the one in these papers to construct an expander overlay to begin with) and reconfigures it into an expander that also optimizes the propagation delay with respect to the underlying geometry. Thus our work can be considered as orthogonal to the above works.

There has been significant amount of work  on a  related problem, namely, constructing distributed hash tables (DHTs) and  associated search protocols that respect the underlying metric~\cite{RD01,prr,karger-ruhl,land,dgs16,satish}. In this line of work, nodes store data items and they can also search for these items. The cost of the search, i.e., the path a request takes from the requesting node to the destination node, is  measured with respect to  an underlying metric. The goal is to build an overlay network and a search algorithm such that the cost of all paths is close to the metric distance.
Our work is broadly in same spirit as these works, with a key difference. While the previous works build an overlay network
while assuming global knowledge of costs between all pairs of nodes, our work assumes that we start with a sparse (expander) topology with
only local knowledge of costs (between neighbors only), which is more realistic in a P2P network.
Furthermore, in these works, the underlying metric is assumed to be \emph{growth-restricted} which is more general than the 2-dimensional
plane assumed here. In a growth-restricted metric, the ball of radius $2r$ around
a point $x$ contains at most a (fixed) constant fraction of points more than the ball of radius $r$ around $x$. This is more general than
the uniform distribution in a 2-dimensional plane assumed here (which  is a special case of growth-restricted) since, in a growth-restricted
metric, points need not be uniformly dense everywhere. An interesting direction of future work is extending our protocols to work in general growth-restricted metrics.

Routing protocols (that are similar in spirit to ours) that assume that each node knows its position and that of its neighbors and that the position of the destination is known to the source are sometimes referred to as geometric routing and greedy approaches to such routing have been explored extensively in the literature (e.g,~\cite{KSU99},~\cite{KK00},~\cite{KWZY03} and the references therein).

\section{\texorpdfstring{$\WEAVER$}{Weaver}: A P2P Topology Construction Protocol}
\label{sec:reconfig-alg}
We show how to convert a given $d$-regular ($d$ is a constant) expander graph embedded in
the Euclidean plane (Definition~\ref{def:embedded})  into a graph that, in addition to having the desired properties of an expander, also allows more efficient
routing and broadcasting with essentially optimal propagation cost.\footnote{As mentioned earlier,
    our protocol will also work for $d$-regular random graphs which are expanders with high probability. Also, the graph need not be regular; it is enough if the degree is bounded, say $O(\polylog n)$, to get the desired performance bounds.} The main result of this section is the $\WEAVER$
protocol, running in $\polylog{n}$ rounds, that yields a network with $\bigO{\log^2 n}$ degree and contains a
\emph{random geometric graph} as a subgraph.

\subsection{The Protocol}
\label{subsec:reconfig-alg}

\noindent \textbf{Brief Description.}
Starting with an embedded, $d$-regular expander $G=(V,E)$, the algorithm constructs a
series of expander graphs, one per phase, such that in each phase $i$, each node $u$
connects to some $\bigO{\log n}$ random neighbors located in a square (box) of side-length $r^i$ centered at $u$ (that intersects the unit square\onlyLong{ --- see Figure~\ref{fig:box-def}}),
where $0 < r < 1$ is a fixed constant (we  can fix $r = 1/4$ due to technical considerations in Section~\ref{subsec:reconfig-analysis-broadcast}).  
In the final phase, $\NUMPHASES$, each node $u$
connects to all $\bigO{\log n}$ neighbors contained in the square of side length $r^{\NUMPHASES}$
at its center. In this manner, we construct a final graph, which is the union of the original graph
and all graphs constructed in each phase, which has low degree ($\bigO{\log^2 n}$) and low
diameter ($\bigO{\log n}$). We note that we require $r^{2\NUMPHASES}n = \bigTh{\log{n}}$,
and hence $\NUMPHASES = \Theta(\log{n} - \log\log{n})/(\log{\sfrac{1}{r}})$.

Our protocol makes extensive use of random walks and
the following lemma is useful in bounding the rounds needed to perform many random walks in parallel under the bandwidth constraints ($\polylog{n}$ bits per edge per round).

\begin{lemma}[Adapted from Lemma 3.2 in~\cite{DNPT13}]\label{lem:random-walk-time}
    Let $G=(V,E)$ be an undirected graph and let each node $v\in V$, with degree $\deg(v)$, initiate  $\eta \deg(v)$ random walks, each
    of length $\lambda$. Then all walks finish their respective $\lambda$  steps in $O(\eta \lambda \log n)$ rounds with high probability.
\end{lemma}

\noindent \textbf{Detailed Description.}
Let $B_u(\ell)$ denote  the \emph{intersection} of the unit square (recall that the Euclidean plane is constrained to a square grid of side length $1$) and the square of side-length
$\ell$ centered at node $u$. Note that if $u$ is located at least distance $\sfrac{\ell}{2}$ from
every edge of the grid, then $B_u(\ell)$ is merely the square with side-length $\ell$ centered
on $u$\onlyLong{ (see Figure~\ref{fig:box-def})}.
\onlyLong{
    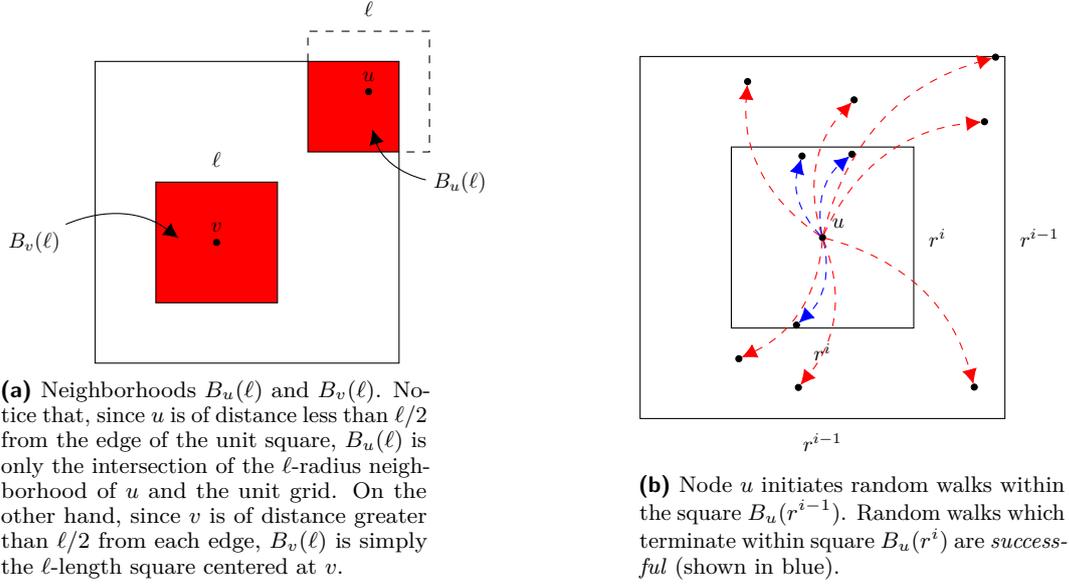
\begin{figure}[t]
        \centering
        \begin{subfigure}{0.4\textwidth}
            \begin{tikzpicture}[scale=.8,transform shape]
                \begin{scope}
                    \clip (0, 0) rectangle (5, 5);
                    \fill[red] (3.5, 3.5) rectangle (5.5, 5.5);
                \end{scope}
                \draw (0, 0) rectangle (5, 5);
                \path (3.5, 3.5) rectangle (5.5, 5.5) node[draw,point,pos=0.5,label=above:{$u$}] {};
                \draw (3.5, 5) |- (5, 3.5);
                \draw[dashed] (3.5, 5) |- (5.5, 5.5) |- (5, 3.5);
                \draw[fill=red] (3, 3) rectangle (1, 1) node[draw,point,fill=black,pos=0.5,label=above:{$v$}] {};

                \path (3.5, 5.5) -- (5.5, 5.5) node[pos=0.5, label={$\ell$}] {};
                \path (3, 3) -- (1, 3) node[pos=0.5, label={$\ell$}] {};

                \node at (4.5, 4) (ulab1) {};
                \node at (6, 3) (ulab0) {$B_u(\ell)$};
                \draw[->] (ulab0) to[bend left] (ulab1);

                \node at (1.5, 2) (vlab1) {};
                \node at (-1, 2) (vlab0) {$B_v(\ell)$};
                \draw[->] (vlab0) to[bend left] (vlab1);
            \end{tikzpicture}
            \caption{ Neighborhoods $B_u(\ell)$ and $B_v(\ell)$. Notice that, since $u$ is of distance less
                than $\sfrac{\ell}{2}$ from the edge of the unit square, $B_u(\ell)$ is only the intersection of the
                $\ell$-radius neighborhood of $u$ and the unit grid. On the other hand, since $v$ is of distance greater
                than $\sfrac{\ell}{2}$ from each edge, $B_v(\ell)$ is simply the $\ell$-length square centered at $v$.}
            \label{fig:box-def}
        \end{subfigure}\hfill
        \begin{subfigure}{0.4\textwidth}
            \begin{tikzpicture}[scale=.8,transform shape]
                \node[draw, point,label=above right:{$u$}] at (0, 0) (u) {};
                \draw (-3, -3) rectangle (3, 3);
                \draw (-1.5, -1.5) rectangle (1.5, 1.5);
                \node[label=below:{$r^{i-1}$}] at (0, -3) {};
                \node[label=below:{$r^i$}] at (0, -1.5) {};
                \node[label=right:{$r^{i-1}$}] at (3, 0) {};
                \node[label=right:{$r^i$}] at (1.5, 0) {};
                \node[draw, point] at (-1.232431637059517, 2.5850316923166483) (rw1) {};
                \node[draw, point] at (-0.42821656730776425, -1.450846794953826) (rw2) {};
                \node[draw, point] at (2.8456353367024887, 2.989201048245599) (rw3) {};
                \node[draw, point] at (0.5194123519714071, 2.2820223265440713) (rw4) {};
                \node[draw, point] at (0.48372687852193774, 1.3811558701202967) (rw5) {};
                \node[draw, point] at (-0.33664236429590666, 1.3509717248940863) (rw6) {};
                \node[draw, point] at (-1.3767613598697954, -2.008930062111596) (rw7) {};
                \node[draw, point] at (-0.3968989831784304, -2.4844671694431995) (rw8) {};
                \node[draw, point] at (2.6636689880523576, 1.919421263856587) (rw9) {};
                \node[draw, point] at (2.4957814747878633, -2.47882020220099) (rw10) {};

                \draw[->, red, dashed] (u) to[bend left] (rw1);
                \draw[->, blue, dashed] (u) to[bend left] (rw2);
                \draw[->, red, dashed] (u) to[bend left] (rw3);
                \draw[->, red, dashed] (u) to[bend left] (rw4);
                \draw[->, blue, dashed] (u) to[bend left] (rw5);
                \draw[->, blue, dashed] (u) to[bend left] (rw6);
                \draw[->, red, dashed] (u) to[bend left] (rw7);
                \draw[->, red, dashed] (u) to[bend left] (rw8);
                \draw[->, red, dashed] (u) to[bend left] (rw9);
                \draw[->, red, dashed] (u) to[bend left] (rw10);
            \end{tikzpicture}
            \caption{Node $u$ initiates random walks within the square $B_u(r^{i-1})$. Random walks which
                terminate within square $B_u(r^i)$ are \emph{successful} (shown in blue).}
            \label{fig:rw-success-failure}
        \end{subfigure}
        \caption{Explanation of $B_u(\ell)$ and demonstration of successful vs. unsuccessful random walks in $B_u(r^i)$.}
        \label{fig:boxexamples}
    \end{figure}
}
Run the following algorithm for $\NUMPHASES = c\log{n}$ phases, for appropriately chosen
constant $c$, starting from phase $1$. The first $\NUMPHASES - 1$ phases are described below and the final phase is described subsequently.

In each phase $1 \leq i \leq \NUMPHASES - 1$, we associate a  graph  with each node $u$ that contains
all nodes and their associated edges inside $B_u(r^i)$  created in  phase $i$   --- which we denote
by $G_u(i)$.   Denote the initial graph for a node $u$ by $G_u(0)$ (note that $G_u(0) \equiv G$). Define $G(i) = \cup_{u \in V} G_u(i)$, i.e., the union of
these graphs across all nodes (note that $G(0) \equiv G$).

\noindent \textbf{First $\NUMPHASES - 1$ phases:} Each phase $i \in \set{1, 2, \dots, \NUMPHASES - 1}$, consists of two major steps outlined below: (Note that we assume at the beginning of phase $i$, graphs $G_u(i-1)$ have been
constructed for all $u$.) In phase $i$, we construct $G_u(i)$ for all $u$ using lazy random walks.

\begin{enumerate}[label=(\textbf{\arabic*}),leftmargin=*]
    \item For each node $u$, perform $\bigTh{\log{n}}$ \emph{lazy random walks}  of length $2\mixingtime$, 
          where $\mixingtime = a\log{n}$ (for a constant $a$ sufficiently large to guarantee rapid mixing, i.e., reaching close
          to the stationary distribution), in $G_u(i-1)$, which is assumed to be an expander (this invariant will be maintained for all $i$).

          A lazy random walk is similar to a normal random walk except that, in each step, the walk
          stays at the current node $u$ with probability $1-\sfrac{\deg(u)}{(\Delta + 1)}$, otherwise it travels to a random neighbor
          of $u$ (in $G_u(i-1)$, i.e., in box $B_u(r^{i-1})$). Here, $\deg(u)$ is the degree of $u$ and $\Delta$ is an upper bound on the maximum degree, which is $O(\log n)$ in $G_u(i-1)$ (by protocol design\onlyLong{ and Lemma~\ref{lem:each-phase-not-too-many-neighbors}}). We maintain the ratio $\deg(u)/(\Delta +1) = O(1)$ in every phase (by protocol design\onlyLong{, Corollary~\ref{cor:lowerboundnumneighbors}, and Lemma~\ref{lem:each-phase-not-too-many-neighbors}}), hence the slowdown of the lazy random walk (compared to the normal random walk) is at most a constant factor.
          It is known  that the stationary distribution of a lazy random
          walk is \emph{uniform} and such a walk,  beginning at a node
          $u$, will arrive at a fixed node $v$ in $G_u(i-1)$ with probability $\sfrac{1}{n} \pm \sfrac{1}{n^3}$ after
          $\mixingtime$ number of steps~\cite{randomwalk}. 
          Thus, a lazy random walk from $u$ gives  a way to sample a node \emph{nearly uniformly} at random from
          the graph $G_u(i-1)$. Each lazy random walk starting from $u$ is represented by a token containing the ID
          of $u$, the current phase number, and the number of steps remaining in the lazy random walk; in phase $i$, this token is passed from node to node to simulate a random walk within $B_u(r^{i-1})$.\footnote{As assumed in Section~\ref{subsec:model}, 
              a node on the random walk path can check whether it is within  the box $B_u(r^{i-1})$ centered at the source node $u$, since it knows the ID Of $u$ and hence the $L_{\infty}$ distance from $u$.}

          Note that each node $v$ that receives the token only considers the subset of its neighbors that are within
          $B_u(r^{i-1})$ when considering nodes to pass the token to. After $2\mixingtime$ steps, if the token lands within
          $B_u(r^i)$\onlyLong{ (see Figure~\ref{fig:rw-success-failure})}, the random walk is \emph{successful}. By Lemma~\ref{lem:random-walk-time}, all walks
          will finish in $O(\log^3 n)$ rounds in the first phase and $O(\log^2 n)$ rounds in subsequent phases with high probability. Note that, to maintain synchronicity, all nodes participate and wait for $O(\log^3 n)$ rounds to finish in the first phase and $O(\log^2{n})$ rounds in subsequent phases.

    \item The graph $G_u(i) = (V_u(i),E_u(i))$ is constructed as follows: its node set $V_u(i)$ is the set
          of all nodes in the box $B_u(r^i)$. Edges from  nodes in $V_u(i)$ are determined  as follows.
          Suppose a lazy random walk from a node $x \in V_u(i)$ successfully ends at $y$, i.e., $y$ is within the box of $B_x(r^i)$ (note that,
          unless $u = x$, this box is different from $B_u(r^i)$, but does overlap with at least \sfrac{1}{4} of $B_u(r^i)$\onlyLong{ --- see Figure~\ref{fig:boxesoverlap})}.
          Node $y$ will send a message to $x$ informing it that its random walk successfully terminated at $y$.
          Among all such nodes  that notify  $x$, $x$ will sample (without replacement) a subset of $b\log{n}$ nodes (for a fixed constant $b$) and add  undirected edges to these sampled nodes.
          The edge set $E_u(i)$ of $G_u(i)$  consists only of edges  between nodes in $V_u(i)$.

\end{enumerate}

\onlyLong{
    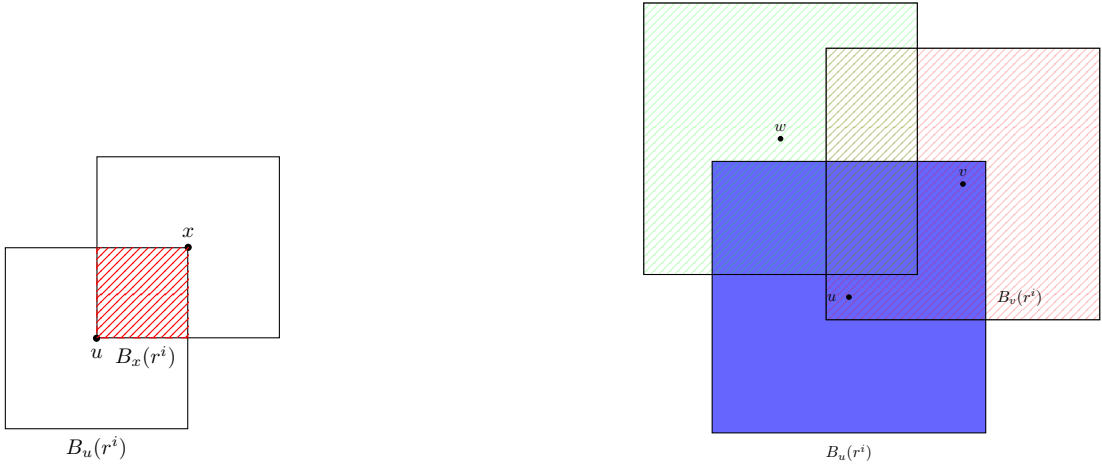
\begin{figure}
        \centering
        \begin{subfigure}[t]{0.4\textwidth}
            \begin{tikzpicture}[scale=0.8,transform shape]
                \node[draw,point,label=below:{$u$}] (u) {};
                \node[draw,rectangle,minimum size=3cm,label=below:$B_u(r^i)$] (bru) at (u) {};
                \node[draw,point,label=above:{$x$}] at (bru.north east) (x) {};
                \node[draw,rectangle,minimum size=3cm,label={[xshift=-0.7cm]below:$B_x(r^i)$}] (brx) at (x) {};
                \begin{scope}
                    \clip (bru.north east) rectangle (bru.south west);
                    \node[draw,fill,dashed,red,pattern=north east lines,pattern color=red,rectangle,minimum size=3cm] () at (x) {};
                \end{scope}
            \end{tikzpicture}
            \caption{Regions $B_u(r^i)$ and $B_x(r^i)$. Notice that, even in the worst case, the overlap is at least \sfrac{1}{4} of the area of the box so long as $x$ falls
                within $B_u(r^i)$.}
            \label{fig:boxesoverlap}
        \end{subfigure}
        \hfill
        \begin{subfigure}[t]{0.4\textwidth}
            \centering
            \begin{tikzpicture}[scale=0.6,transform shape]
                \fill[blue,opacity=0.6] (-3, -3) rectangle (3, 3);
                \draw[pattern=north east lines, pattern color=red,opacity=0.5] (-0.5, -0.5) rectangle (5.5, 5.5);
                \draw[pattern=north east lines, pattern color=green,opacity=0.5] (-4.5, 0.5) rectangle (1.5, 6.5);
                \draw (-3, -3) rectangle (3, 3) node[pos=0.5,point,draw,label={[xshift=-0.1cm]left:$u$}] {};
                \node[label=below:{$B_u(r^i)$}] at (0, -3) {};
                \node[label=right:{$B_v(r^i)$}] at (3, 0) {};

                \draw (-0.5, -0.5) rectangle (5.5, 5.5) node[pos=0.5,point,draw,label=above:{$v$}] (v) {};
                \draw (-4.5, 0.5) rectangle (1.5, 6.5) node[pos=0.5,point,draw,label=above:{$w$}] (w) {};
            \end{tikzpicture}
            \caption{Rectangles of side-length $r^i$ centered at nodes $u$, $v$, and $w$. Notice that,
                since $v$ falls within $B_u(r^i)$, a random walk from $v$ can connect to $u$. On the other hand,
                a random walk from $w$ cannot connect to $u$.}
            \label{fig:in-edges}
        \end{subfigure}
        \caption{Demonstrating the overlap of boxes centered at multiple nodes.}
    \end{figure}
}

\noindent \textbf{Last phase:}
The final phase is similar, except that each node $u$ initiates a larger number of random walks, so that with high probability
\emph{all nodes} within the box $B_u(r^{\NUMPHASES})$ (note that $r^{\NUMPHASES} = \Theta(\sqrt{\log n/n})$)  are sampled and thus $u$ is able to form connections
to all  nodes  in $B_u(r^{\NUMPHASES})$ (which contains $\Theta(\log n)$ nodes). This will ensure that $G(\NUMPHASES) = \cup_{u \in V} G_u(\NUMPHASES)$ contains a random geometric graph $G(n,\rho)$ with $\rho = \Theta(\sqrt{\log n/n})$.
More precisely, in the final phase (phase $\NUMPHASES$), each node $u$ runs $\Theta(\log^2 n)$ random walks on $G(\NUMPHASES-1)$ 
to all nodes within $B_u(r^{\NUMPHASES})$ to form graph $G(\NUMPHASES)$.

The final graph $G^*$ is the union of the graphs $G(i)$, $0 \leq i \leq \NUMPHASES$.
Algorithm~\ref{alg:reconfig} gives a high-level summary of the protocol.

\begin{algorithm}[h]
    \footnotesize
    \caption{$\WEAVER$ Construction Protocol}
    \label{alg:reconfig}
    \begin{algorithmic}[1]
    \ForEach{node $u$ and phase $i$ in $\set{1, 2, \dots, \NUMPHASES - 1}$}
        \State $u$ performs $\Theta(\log{n})$ random walks of length $2\mixingtime = \Theta(\log{n})$ in $G_u(i - 1)$
        \State $u$ connects to $\Theta(\log{n})$ nodes where random walks are successful
    \EndForEach
    \ForEach{node $u$}
        \State $u$ initiates $\Theta(\log^2{n})$ lazy random walks in $B_u(r^{\NUMPHASES})$ and connects to nodes where walk ends successfully
    \EndForEach
    \end{algorithmic}
\end{algorithm}

\subsection{Protocol Analysis}
\label{subsec:reconfig-analysis-properties}

We  prove that, with high probability, the  protocol takes $O(\log^3 n)$ rounds and constructs a graph $G^{*}$  that  has maximum degree $O(\log^2 n)$ and contains a random geometric graph as a subgraph (besides being an expander).

To argue that the constructed graph $G^*$ contains a random geometric graph, we show that the series of graph constructions proceeds correctly in each phase, resulting in the last phase constructing $G({\NUMPHASES})$, the desired random geometric graph. Each phase $i$ crucially relies on the fact that the subgraph induced by a given node $u$ in phase $i-1$, $G_u(i-1)$, is an expander\onlyLong{, and this is shown in Lemma~\ref{lem:graph-in-each-phase-expander}}. In each phase $i$, we perform several lazy random walks starting from each node $u$ on $G_u(i-1)$. Since the lengths of lazy random walks starting at $u$ performed on $G_u(i-1)$ are $\Omega(\log n)$, we see that they run longer than the mixing time of
$G_u(i-1)$, resulting in the final destination of the walk, i.e, the neighbor of $u$ in that phase resulting from the random walk, being chosen uniformly at random from the vertices of $G_u(i-1)$. This property is useful \onlyShort{in the analysis}\onlyLong{when proving Lemma~\ref{lem:graph-in-each-phase-expander}}.
The random walks performed
by each node $u$ in phase $i$ result in at least $\bigOm{\log n}$ neighbors that
can be used by $u$ to construct its part of the graph $G(i)$ \onlyLong{(Lemma~\ref{lem:each-phase-at-least-enough-neighbors})}.
Finally, with high probability, the subgraph induced by edges of length less than $r^{\NUMPHASES}$ forms a random geometric graph(\onlyShort{see full version}\onlyLong{Lemma~\ref{lem:all-nodes-connected-final-phase}}).

To argue about the maximum degree of the graph, notice first that we construct at most $\NUMPHASES = O(\log n)$ subgraphs, one per phase. By showing that each of these subgraphs has a maximum degree of $O(\log n)$ with high probability, we show that the maximum degree of the graph is $O(\log^2 n)$ with high probability.
\onlyLong{Lemma~\ref{lem:each-phase-not-too-many-neighbors} shows that the}\onlyShort{The} degree of any node
in $G(i)$ does not exceed $O(\log n)$, for all phases $i$ excluding the final
phase $\NUMPHASES$ \onlyShort{(see full paper)}. \onlyLong{Lemma~\ref{lem:min-max-bounds-nodes-space} is a general lemma that bounds}\onlyShort{In our analysis, we bound} the number of nodes in a box surrounding a given node, and in particular \onlyLong{it shows}\onlyShort{show} that the degree of each node in the final phase  does
not exceed $\bigO{\log n}$ with high probability, i.e., the maximum degree of the graph
$G(\NUMPHASES)$ is $\bigO{\log n}$.

All these properties of the final graph $G^*$ are captured by Theorem~\ref{thm:graph-conversion-alg-works}. We argue about the run time directly in the proof of Theorem~\ref{thm:graph-conversion-alg-works}.

The key lemma is showing  that each graph $G(i)$ formed at the end of each phase $i$ is an expander \onlyLong{(Lemma~\ref{lem:graph-in-each-phase-expander})}.
It can be proved by induction on $i$. The base case is given, since $G_u(0) \equiv G$ and $G$ is an expander.
For the induction hypothesis, we assume that $G_u(i-1)$ is an expander and prove that $G_u(i)$ is an expander as well. 
The main technical idea behind the proof is to show that, with high probability, every subset of nodes (that is of size at most half the size of $V_u(i)$)  has a conductance that is at least some fixed constant.
The protocol initiates random-walks by each node in each phase of the algorithm to construct an expander, and the random walks occur over different subgraphs (regions). This make it  non-trivial to show that the constructed subgraph around each node is an expander at each phase. Since each node does random walks in a local region around itself, the expansion proof has to be done carefully. \onlyShort{To save space, we leave the required lemmas and proof of Theorem~\ref{thm:graph-conversion-alg-works} for the full version of the paper.}
\onlyLong{%
    \begin{lemma}\label{lem:graph-in-each-phase-expander}
        Each $G_u(i)$, for all $u$, and for all $0\leq i \leq \NUMPHASES-1$, is an expander with high probability.
    \end{lemma}

    \begin{proof}
        We proceed by induction on $i$. The base case is given, since $G_u(0) \equiv G$  and $G$ is an expander.
        Fix a node $u$. For the induction hypothesis, we assume $G_u(i-1)$ is an expander and show that $G_u(i)$ is also an expander with high probability. (By union bound, this will hold for all $u$ with high probability.)

        We show that, with high probability,
        every subset $S$ of $V_u(i)$ of size at most $\sfrac{|V_u(i)|}{2}$ has $\bigOm{1}$ conductance. Recall that $V_u(i)$ denotes the set of all nodes in $B_u(r^i)$. The proof uses a similar approach to that of Lemma 1 in \cite{focs2015}, however, the proof here is more involved since each node initiates random walks in a different underlying graph.

        From Lemma~\ref{lem:min-max-bounds-nodes-space}, for all $u$ and $i$, the number of nodes in $B_u(r^i)$ is $\bigTh{r^{2i}n\log{n}}$ with high probability (hence $|V_u(i)| = \bigTh{r^{2i}n\log{n}}$ w.h.p.).

        Let $n_u(i) = \card{V_u(i)}$ (i.e., $G_u(i)$ has $n_u(i)$ nodes) and let $S$ be any subset of $V_u(i)$ of size at most $\sfrac{n_u(i)}{2}$. Let $|S| = s$.  Recall that $E(S,\comp{S})$ is the set of edges with one endpoint in $S$ and the other
        in $\comp{S}$. We wish to show that
        \[\card{E(S, \comp{S})} \geq c_1\card{S}\log n\]
        for some constant $c_1$ (defined later in the analysis). Recall that the $\WEAVER$ protocol maintains a degree upper bound of
        $\bigO{\log n}$ for all nodes in $G_u(i)$ and for all $i$ (see Lemma~\ref{lem:each-phase-not-too-many-neighbors}). Thus
        \[\vol(S) = \sum_{x \in S}\deg(v) \leq \card{S}\max\limits_{x\in S}\left(\deg(v)\right) \leq c_2\card{S}\log{n}\]
        for some constant $c_2$, and we have
        \[\frac{1}{\vol(S)} \geq \frac{1}{c_2\card{S}\log{n}}.\]
        Thus, if $\card{E(S, \comp{S})} \geq c_1\card{S}\log n$, then $G_u(i)$ has conductance
        \[\frac{\card{E(S, \comp{S})}}{\vol(S)} \geq \frac{c_1\card{S}\log n}{c_2\card{S}\log{n}} = \frac{c_1}{c_2} = \Theta(1).\]

        Now, to show that $\card{E(S, \comp{S})} \geq c_1\card{S}\log n$ with high probability, we consider the event that
        $\card{E(S,\overline{S})}$ is greater than or equal to $c_1\card{S}\log{n}$. The $\WEAVER$ protocol ensures
        that every node $x \in S$ establishes at most $d = c\log n$ edges, with high probability, to other
        nodes using lazy random walks starting at $x$  for some constant
        $c > 0$ (Lemma~\ref{lem:each-phase-at-least-enough-neighbors}).\footnote{Note that other nodes might connect to $x$ via random walks, we don't consider these edges here, as these are at most a constant factor of the edges of the overall degree of $x$ (by Lemma \ref{lem:each-phase-not-too-many-neighbors}). These edges will affect
            only the conductance by a constant factor.} Note that the random walks initiated from each node $x$ will be within the box $B_x(r^i)$.

        Consider a node $x \in S$.
        Let $E(x)$ denote the set of edges incident on node $x$ that are formed using lazy random walks from $x$. Let
        \[E'(S) = \bigcup_{x \in S}E(x).\]
        Notice that $E'(S) \subseteq E(S, S) \cup E(S, \comp{S})$.  We will show that sufficiently many of the edges in $E'(S)$ belong
        to the set $E(S,\comp{S})$ with high probability. To prove this, we upper bound the (complementary) probability that a large fraction of $E'(S)$
        belong to $E(S,S)$.

        For a node $x \in S$, since the edges in $E(x)$   are established via (lazy) random walks in  $V_x(i-1)$ (which are the nodes in the box $B_x(r^{i-1})$), the probability that a random walk ends in $S$ (which is a subset of $V_u(i)$ and lies within box $B_u(r^i)$) depends on the intersection of $S$ with $V_x(i)$. It can be shown by geometric
        arguments that for at least a constant fraction of nodes in $S$ (we call these  {\em good} nodes), at least a constant fraction of its (respective)
        random walks will end outside $S$ {\em and} in $B_u(r^i)$ (see Figure~\ref{fig:overlappingsubregions}). This is because of the following reasons: (1) for any $x \in S$,  $B_x(r^i)$ and $B_u(r^i)$ overlap each other in a constant fraction of their areas; (2)  the number of nodes in both these boxes is (concentrated at)
        $\bigTh{n_{i-1}}$ with high probability (by Lemma~\ref{lem:min-max-bounds-nodes-space});  and (3) the random walks
        are nearly uniform in $V_x(i-1)$.
        Let the set of good nodes $x \in S$ be denoted by $\operatorname{Good}(S)$. Note that $|\operatorname{Good}(S)| = \Theta(|S|) = \Theta(s)$.
        Hence for all  nodes $x \in \operatorname{Good}(S)$, a specific edge $e \in E(x)$ has a particular endpoint in $S$ (which is a subset of $V_u(i)$ and lies within
        box $B_u(r^i)$),
        say $z$, is simply the probability that
        the corresponding random walk terminated at $z$, which is $\Theta(\sfrac{1}{|V_x(i-1)|}) = \Theta(\sfrac{1}{n_{i-1}})$
        where, $n_{i-1} = \Theta(r^{2(i-1)}n\log{n})$.

        Choose a sufficiently small constant $c_1 < c$, and set $q = c_1 \log n$ and $\gamma = d - q = (c - c_1)\log n$.
        Letting $X_j$ denote the indicator random variable that $e_j \in E'(\operatorname{Good}(S))$ has both
        end points in $S$, we have
        $\prob{X_j = 1} \leq \Theta(\frac{s}{n_{i-1}})$
        by union bound across all good nodes in $S$.
        Observe then that
        $\prob{\land_{j=1}^q X_j = 1} = (\sfrac{s}{n_{i - 1}})^q$.

        Now, the probability that a particular set $S$ of size $s$ does not satisfy $E(S, \comp{S}) \geq c_1\log{n}\card{S}$
        is upper bounded by  the probability that  $\card{E(S, S)} \geq \gamma\card{S}$.
        There are $\binom{ds}{\gamma s}$ such ways to choose $\gamma s$ edges and $\binom{n_{i - 1}}{s}$
        sets $S$ of size $s$, hence
        \begin{align*}\prob{\exists S \suchthat \card{S} = s\text{ and }\card{E(S, S)} \geq \gamma\card{S}}
             & \leq \binom{n_{i - 1}}{s}\binom{ds}{\gamma s}\left(\frac{s}{n_{i - 1}}\right)^{\gamma s}.
        \end{align*}
        We upper bound the above probability and show that it is small for all subsets $S$  of size at most $n_i/2$.
        We consider two cases, when $s=|S| = \littleO{n_{i - 1}}$ and $s=|S| = \bigOm{n_{i - 1}}$.
        \begin{enumerate}[itemindent=2em,label=Case \roman*:]
            \item $s = \littleO{n_{i-1}}$

                  Apply the inequality $\binom{n}{q} \leq \left(\frac{en}{q}\right)^q$. The probability
                  that there exists a set $S$ of size $s$ such that $\card{E(S, S)} \geq \gamma \card{S}$
                  is bounded above by
                  \begin{align*}\binom{n_{i-1}}{s}\binom{ds}{\gamma s}\left(\frac{s}{n_{i-1}}\right)^{s\gamma}
                       & \leq \left(\frac{en_{i - 1}}{s}\right)^s\left(\frac{eds}{\gamma s}\right)^{\gamma s}\left(\frac{s}{n_{i-1}}\right)^{\gamma s} \\
                       & =\left(\frac{e}{s}\right)^s s^{\gamma s} \left(\frac{de}{\gamma}\right)^{\gamma s}n_{i-1}^{s(1-\gamma)}                       \\
                       & =e^s s^{s(\gamma - 1)}n^{s(1-\gamma)}\left(\frac{ed}{\gamma}\right)^{\gamma s}
                  \end{align*}
                  Setting $\beta = \ln\left(\frac{de}{\gamma}\right)$, this is
                  \begin{align*}
                       & \exp\left(s + s(\gamma - 1)\ln{s} + s(1 - \gamma)\ln{n_{i - 1}} + \beta\gamma s\right) \\
                       & =\exp\left(s\left(1-(\gamma - 1)\ln\frac{n_{i-1}}{s} + \beta\gamma\right)\right)       \\
                       & \leq \frac{1}{n_{i-1}^{\Omega(\log n)}} \leq 1/n^3
                  \end{align*}
                  since $s = \littleO{n_{i-1}}$ and $\gamma = \Theta(\log n)$, and $n_{i-1} = \Omega(\log n)$.
            \item $s \geq \zeta n_{i-1}$ for some $\zeta > 0$

                  For sufficiently large $n$, we can apply Stirling's Formula to deduce:
                  $\binom{n}{q} \leq \exp\left(nH(\sfrac{q}{n})\right)$.
                  where $H(\varepsilon) = -\varepsilon\ln\varepsilon - (1-\varepsilon)\ln(1-\varepsilon)$ is the
                  entropy of $\varepsilon$. In particular, since $\zeta n_{i-1} \leq s \leq \sfrac{1}{2}n_{i-1}$, for some $\zeta > 0$,
                  the probability that there exists a set $S$ of size $s$ such that $\card{E(S, S)} \geq \gamma \card{S}$
                  is bounded above by
                  \[\binom{n_{i-1}}{s}\binom{ds}{\gamma s}\left(\frac{s}{n_{i-1}}\right)^{\gamma s}\]
                  Setting $\beta = \frac{s}{n_{i-1}}$, this is                                                     \\
                  \[\exp\left(n_{i-1}H\left(\frac{s}{n_{i-1}}\right)+dsH\left(\frac{\gamma s}{ds}\right)+\ln(\beta)\gamma s\right)\]
                  which is bounded above by
                  \[\exp\left(n_{i-1}\left(H(\zeta) + \zeta dH\left(\frac{\gamma}{d}\right) + \ln(\beta)\gamma\zeta\right)\right)\]
                  Now, the expression $H(\zeta) + \zeta dH\left(\frac{\gamma}{d}\right) + \ln(\beta)\gamma\zeta$
                  is $-\bigTh{\log{n}}$ since $n_{i-1} = \bigOm{\log{n}}$, and, since $d$ and $\gamma$ are $\bigOm{\log{n}}$, the previous expression is upper bounded by
                  \[\exp\left(-\bigTh{n_{i-1}}\right) \leq \frac{1}{n^3}.\]
        \end{enumerate}
        Thus we have shown that in $G_u(i)$, for all subsets $S$  of size at most $n_i/2$, $\Pr(E(S,\overline{S})) \geq (d-q)\log n |S|$ with probability at least $1-1/n^3$. Hence  $G_u(i)$ is an expander with probability at least $1 -1/n^3$. Hence by union bound, each $G_u(i)$ is an expander
        for all $u \in V$ and all $1 \leq i \leq \NUMPHASES-1$.
    \end{proof}

    \begin{figure}
        \centering
        \begin{tikzpicture}[scale=1.2,transform shape]
            \node[draw,point] (u) {};
            \begin{scope}
                \node[draw,rectangle,minimum size=3cm,label=below:$B_u(r^i)$] (bru) at (u) {};
                \clip (bru.north east) rectangle (bru.south west);
                \node[draw,point,above right=1.2584221194346479cm and 1.2126608834590886cm of u,blue] (s1) {};
                \node[draw,point,above right=0.8291039217772134cm and 0.5594547915183995cm of u,red] (s2) {};
                \node[draw,point,above right=1.0873286600807703cm and 1.1721755690999265cm of u,blue] (s3) {};
                \node[draw,point,above right=0.3863358750316113cm and 0.52768544298164cm of u,red] (s4) {};
                \node[draw,point,above right=0.7804370957569715cm and -1.0574765355969202cm of u,red] (s5) {};
                \node[draw,point,above right=0.03288786189035349cm and 0.8946699534577767cm of u,red] (s6) {};
                \node[draw,point,above right=1.233456444220374cm and 0.2391640660989146cm of u,red] (s7) {};
                \node[draw,point,above right=1.0704500887423762cm and -0.5589391477142992cm of u,black] (s8) {};
                \node[draw,point,above right=0.5741510659059312cm and -1.2029840332827666cm of u,black] (s9) {};
                \node[draw,point,above right=-0.3698352475793866cm and -1.2929600716092091cm of u,black] (s10) {};
                \node[draw,point,above right=-0.6991223702920697cm and -0.4450625068017335cm of u,black] (s11) {};
                \node[draw,point,above right=-1.2727845262665987cm and -0.363648399675642cm of u,black] (s12) {};
                \node[draw,point,above right=-0.8692044835399281cm and 1.0986334737285351cm of u,black] (s13) {};
                \node[draw,point,above right=-1.0161788588167657cm and -1.00981438861632cm of u,black] (s14) {};
                \node[draw,point,above right=-0.6240602229223468cm and -1.1487821006284962cm of u,black] (s15) {};
                \node[draw,point,above right=-1.2258256054651735cm and 1.220158172719026cm of u,black] (s16) {};
                \node[draw,point,above right=0.7703341825842305cm and -0.7375706933157724cm of u,black] (s17) {};
                \node[draw,point,above right=1.2949718351809087cm and -1.1097079593424495cm of u,black] (s18) {};
                \node[draw,point,above right=-0.5435264645034769cm and 1.0208062323809583cm of u,black] (s19) {};
                \node[draw,point,above right=-0.399038283246017cm and -1.1634733293097979cm of u,black] (s20) {};
                \node[draw,point,above right=1.1553070031631087cm and 0.902675341865695cm of u,black] (s21) {};
                \node[draw,point,above right=-0.048277813603857016cm and 0.2690418244061282cm of u,black] (s22) {};
                \node[draw,point,above right=-1.111307949210607cm and 0.7138311350716008cm of u,black] (s23) {};
                \node[draw,point,above right=0.27398811624899355cm and 0.47896216545140485cm of u,red] (s24) {};
                \node[draw,point,above right=1.0341865036892306cm and -0.27831434894056534cm of u,red] (s25) {};

                \node[draw,dashed,rectangle,minimum size=3cm,blue] (brs1) at (s1) {};
                \node[draw,dashed,rectangle,minimum size=3cm,blue] (brs1) at (s3) {};

                \draw[->,blue,dashed] (s1) to[bend right] (s21);
                \draw[->,blue,dashed] (s1) to[bend right] (s22);
                \draw[->,blue,dashed] (s1) to[bend right] (u);
                \draw[->,blue,dashed] (s3) to[bend left] (s21);
                \draw[->,blue,dashed] (s3) to[bend left] (s22);
                \draw[->,blue,dashed] (s3) to[bend left] (u);

                \draw[->,red,dashed] (s4) to[bend right] (s8);
                \draw[->,red,dashed] (s4) to[bend right] (s11);
                \draw[->,red,dashed] (s4) to[bend right] (s13);
                \draw[->,red,dashed] (s4) to[bend right] (s17);
                \draw[->,red,dashed] (s4) to[bend right] (s19);
                \draw[->,red,dashed] (s4) to[bend right] (s21);
                \draw[->,red,dashed] (s4) to[bend right] (s22);
                \draw[->,red,dashed] (s4) to[bend right] (u);
                \node[draw,dashed,rectangle,minimum size=3cm,red] (brs1) at (s4) {};

            \end{scope}
        \end{tikzpicture}
        \caption{Blue and red nodes denote nodes in $S$. Notice that, while the blue nodes only have 3 edges leaving $S$, the red nodes have many edges
            leaving $S$. These red nodes correspond to the constant factor of $S$ contributing to the high expansion of $G_u(i)$.}
        \label{fig:overlappingsubregions}
    \end{figure}
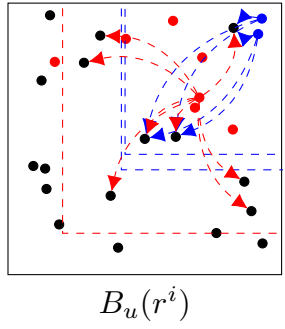

    Throughout this analysis, we say that a random walk from a node $u$ in phase $i$ is \emph{successful} if
    the random walk finishes inside of the inner box, i.e., in $B_u(r^i)$.

    \begin{lemma}\label{lem:each-phase-at-least-enough-neighbors}
        In each phase $i$, for each node $u$, at least $72\log{n}$ of the random
        walks lead to nodes within $B_u(r^i)$ with probability at least $1 - \sfrac{1}{(n^2\NUMPHASES)}$.
    \end{lemma}

    \begin{proof}
        By construction, each random walk must stop at a node in $B_u(r^{i-1})$. In
        particular, then, the probability that a walk starting at $u$ finishes in $B_u(r^i)$ is
        \[\frac{r^i \times r^i}{r^{i-1} \times r^{i-1}} = r^2\]
        since the nodes are distributed uniformly at random within the unit square and the random walks are run for the mixing time of the underlying induced previous phase graph. Note that since we fix the length of the random walk at $O(\log n)$ and the previous phase graphs are expanders from Lemma~\ref{lem:graph-in-each-phase-expander}, we see that these walks are run for mixing time.

        Now, set
        \[X_j = \begin{cases}1 & \hbox{ walk $j$ lands in $B_u(r^i)$} \\
             0 & \hbox{ otherwise}\end{cases}\]
        and
        \[X = \sum_{j=1}^d X_j\]
        where $d$ is the number of random walks emanating from $u$.
        Notice that $\expectation{X} = dr^2$. Since $d = \bigO{\log{n}}$, say $d \leq c \log{n}$ for
        sufficiently large $c$, by a simple application of Chernoff bounds, we have
        \begin{align*}\prob{X < (1 - \delta)dr^2}
             & \leq \prob{X < (1 - \delta)2cr^2\ln{n}}                                                                                     \\
             & \leq \exp\left(-\frac{\delta^2 2cr^2\ln{n}}{2}\right)                                                                       \\
             & = n^{-\delta^2cr^2}                                                                                                         \\
             & < \frac{1}{n^{2 + \varepsilon}}\hbox{ for $\delta > \frac{1}{r}\sqrt{\frac{2 + \varepsilon}{c}}$ and any $\varepsilon > 0$} \\
             & < \frac{1}{n^2\NUMPHASES}\hbox{ since $\NUMPHASES = \bigO{\log{n}}$}
        \end{align*}
        In particular, taking \[\delta = 1 - \frac{36}{cr^2\ln{2}}\] yields 
        \[\prob{X < 72\log{n}} < \frac{1}{n^2\NUMPHASES}\] as desired.
    \end{proof}

    \begin{corollary}\label{cor:lowerboundnumneighbors}
        With probability at least $1 - \sfrac{1}{n} - \sfrac{1}{n^c}$ for some (arbitrary but fixed)
        $c > 0$, every node initiates $\bigOm{\log{n}}$ successful random walks in each phase.
    \end{corollary}

    \begin{proof}
        By a union bound across all $n$ nodes and $\NUMPHASES$ phases, the probability
        that any node in any phase completes $\littleO{\log{n}}$ successful random walks
        is bounded above by $\sfrac{1}{n} + \sfrac{1}{n^c}$. 
        Thus, with probability at least
        $1 - \sfrac{1}{n} - \sfrac{1}{n^c}$, every node initiates $\bigOm{\log{n}}$ successful
        random walks in each phase.
    \end{proof}

    In order to bound the degree of each node in each phase of the algorithm in Lemma~\ref{lem:each-phase-not-too-many-neighbors}, we make use of the following bound on the number of nodes that occupy a given region in Lemma~\ref{lem:node-concentration}.

    \begin{lemma}\label{lem:node-concentration}
        Any region $U$ with area $A \geq \sfrac{c\ln{n}}{n}$, for a sufficiently large constant $c$, contains $\bigTh{nA}$ nodes with high probability.
    \end{lemma}

    \begin{proof}
        Let $X_v$ denote the indicator random variable that is 1 when $v$ falls in $U$. Clearly,
        $\expectation{X_v} = A$ and
        \[X = \sum_{v\in V} X_v\]
        has expectation $nA$. For any $0 \leq \delta \leq 1$, we have
        \begin{align*}\prob{X \leq (1 - \delta)nA}
             & \leq \exp\left(\sfrac{-\delta^2nA}{2}\right)                                               \\
             & \leq \exp\left(\sfrac{-c\delta^2\ln{n}}{2}\right)\hbox{ since $A \geq \sfrac{c\ln{n}}{n}$} \\
             & =n^{\sfrac{-c\delta^2}{2}}
        \end{align*}
        Similarly, we have
        \begin{align*}\prob{X \geq (1 + \delta)nA}
             & \leq \exp\left(\sfrac{-\delta^2nA}{3}\right)                                               \\
             & \leq \exp\left(\sfrac{-c\delta^2\ln{n}}{3}\right)\hbox{ since $A \geq \sfrac{c\ln{n}}{n}$} \\
             & =n^{\sfrac{-c\delta^2}{3}}
        \end{align*}
        Thus, we have
        \begin{align*}\prob{(1 - \delta)nA \leq X \leq (1 + \delta)nA}
            =\prob{X \leq (1 - \delta)nA \text{ or } X \geq (1 + \delta)nA}                         \\
            \leq \prob{X \leq (1 - \delta)nA} + \prob{X \geq (1 + \delta)nA} \hbox{ by union bound} \\
            \leq n^{\sfrac{-c\delta^2}{2}} + n^{\sfrac{-c\delta^2}{3}}
        \end{align*}
        as desired.
    \end{proof}
 
    The following lemma establishes that a node
    will not receive too many ``incoming'' connections, i.e., that, with high
    probability, the number of nodes $v$ whose random walks successfully end at
    node $u$ is $\bigO{\log{n}}$ across all nodes $u$. (Figure~\ref{fig:in-edges} shows an example of a node which is close enough to $u$ so that its walks can (possibly) connect to $u$ and a node which is too far away.)

    \begin{lemma}\label{lem:each-phase-not-too-many-neighbors}
        Fix a node $u$. During phase $i$, the number of nodes that
        connect to $u$ is $\bigO{\log{n}}$ with probability $1 - \sfrac{1}{n^3}$.
    \end{lemma}

    \begin{proof}
        Let $B_u$ denote the neighborhood of $u$ during this phase. Write
        \[X_i^v = \begin{cases}1 & \hbox{ the $i$-th random walk from $v$ lands in $u$}\\0 & \hbox{ otherwise}\end{cases}\]
        and observe that $X^v = \sum_{i = 1}^d X_i^v$, where $d = \bigO{\log{n}}$ is the number of random walks initiated by $u$,
        upper bounds the indicator random variable of whether $v$ connects
        to $u$. Then
        \[\expectation{X_i} \leq \frac{1}{N_v} + \frac{1}{N_v^c} \hbox{ for an arbitrary, but fixed, constant $c \geq 1$}\] 
        where $N_v$ is the number of nodes in the neighborhood around $v$. Thus, we have $\expectation{X} \leq
            \sfrac{d}{N_v} + \sfrac{d}{N_v^c}$. Similarly, the random variable
        \[X = \sum_{v \neq u}X^v\]
        upper bounds the number of incoming connections to $u$. Note that if $v \notin B_u$, then
        $X^v$ is identically 0. Thus, by linearity of expectation,
        \begin{align*}\expectation{X}
             & = \sum_{\genfrac{}{}{0pt}{1}{v \neq u}{v \in B_u}}\left(\frac{d}{N_v} + \frac{d}{N_v^c}\right)                                                                                    \\
             & = \sum_{\genfrac{}{}{0pt}{1}{v \neq u}{v \in B_u}}\left(\frac{d}{\bigTh{r^{2i}n}} + \frac{d}{\bigTh{r^{c\left(2i\right)}n^c}}\right)\hbox{ by Lemma~\ref{lem:node-concentration}} \\
             & = \bigTh{r^{2i}n}\left(\frac{d}{\bigTh{r^{2i}n}} + \frac{d}{\bigTh{r^{c\left(2i\right)}n^c}}\right)\hbox{ again, by Lemma~\ref{lem:node-concentration}}                           \\
             & = \bigO{d}
        \end{align*}

        Now, applying Chernoff Bounds, noting that $d = \bigO{\log{n}}$,
        \begin{align*}\prob{X \geq (1 + \delta)\expectation{X}}
             & \leq \exp\left(\frac{-\delta^2\expectation{X}}{3}\right)                   \\
             & \leq \exp\left(\frac{-\delta^2k\log{n}}{3}\right) \hbox{ for some $k > 0$} \\
             & \leq \frac{1}{n^3}\hbox{ for appropriate choice of $\delta$}
        \end{align*}
        Thus, $X = \bigOm{\log{n}}$ with probability at least $1 - \sfrac{1}{n^3}$.
    \end{proof}

    We now bound the number of nodes in the neighborhood $B_u(r^i)$. The proof of the following lemma follows from Lemma~\ref{lem:node-concentration}.

    \begin{lemma}\label{lem:min-max-bounds-nodes-space}
        With high probability, for all $0 \leq i \leq \NUMPHASES$ and all nodes $u$
        in $G$, the number of nodes in $B_u(r^i)$ is $\bigTh{r^{2i}n\log{n}}$ which is $\bigOm{\log{n}}$ and
        $\bigO{r^{2i}n\log{n}}$.
    \end{lemma}

    \begin{proof}
        Notice that the area of the region $B_u(r^i)$, $r^{2i}$, is bounded below
        by $r^{2\NUMPHASES} = \bigTh{\sfrac{\log{n}}{n}}$. The result therefore
        follows directly from Lemma~\ref{lem:node-concentration}: the number of nodes
        in $B_u(r^i)$ is $\bigTh{r^{2i}n\log{n}}$ with high probability, which is $\bigO{r^{2i}n\log{n}}$
        and $\bigOm{\log{n}}$.
    \end{proof}

    The proof of the following lemma can be seen from using Lemma~\ref{lem:min-max-bounds-nodes-space} and basic probability arguments.
    \begin{lemma}\label{lem:all-nodes-connected-final-phase}
        At the end of Algorithm~\ref{alg:reconfig}, $u$ is connected to all nodes in the
        graph $G_u(\NUMPHASES)$ with probability $1 - \sfrac{1}{n^c}$ for some $c >
            0$.
    \end{lemma}

    \begin{proof}
        We will upper bound the probability that a particular node $v$ in $B_v(r^{\NUMPHASES})$
        is not a neighbor of $u$. By Lemma~\ref{lem:min-max-bounds-nodes-space}, there are $\bigO{r^{2\NUMPHASES}n\log{n}} =
            \bigO{\log{n}}$ nodes in this square with high probability and, since $u$
        does $\bigO{\log^2{n}}$ random walks, the probability that $u$ does not
        connect to $v$ is
        \[\left(1 - \frac{1}{\log{n}}\right)^{\bigO{\log^2{n}}} \leq \frac{1}{n^c}\hbox{ for some $c > 0$}\] 
        Taking a union bound across all $O(\log{n})$ nodes completes the proof.
    \end{proof}

}

\begin{theorem}\label{thm:graph-conversion-alg-works}
    The $\WEAVER$ protocol (Algorithm \ref{alg:reconfig}) takes an embedded $d$-regular expander graph and constructs a graph in $O(\log^3 n)$ rounds such that:
    \begin{enumerate}[label=(\roman*)]
        \item its degree is $O(\log^2 n)$ with high probability \label{item:graphconv-deg}
        \item and it contains a random geometric graph $G(n,\rho)$ (where $\rho = \Theta(\sqrt{(\log n)/n}$) with high probability. \label{item:graphconv-gnr}
    \end{enumerate}

\end{theorem}

\onlyLong{
    \begin{proof}
        By the construction of our algorithm and by Lemma~\ref{lem:each-phase-not-too-many-neighbors}, the number of connections established per
        phase is $\bigO{\log{n}}$ with high probability. Since there are $\NUMPHASES = \bigO{\log{n}}$ phases, item~\ref{item:graphconv-deg} directly follows. Item~\ref{item:graphconv-gnr}
        directly follows from Lemma~\ref{lem:all-nodes-connected-final-phase}.

        To argue about run time, notice that there are $\NUMPHASES = \bigO{\log n}$ phases. In each of the first $\NUMPHASES - 1$ phases, each node performs a $\bigO{\log n}$ lazy random walks simultaneously, each of length $\bigO{\log n}$.

        We first argue that the time it takes to perform lazy random walk in each phase is asymptotically equal to the time it takes to perform a random walk.  Recall that in phase $i$, the lazy random walk starting at node $u$ is run on graph $G_u(i-1)$. Let $d_u$ be the degree of node $u$ in $G_u(i-1)$ and let $\Delta$ be (an upper bound on) the maximum degree of any node in $G_u(i-1)$. Since we maintain the ratio $d_u/(\Delta +1) = O(1)$ in every phase (by algorithm construction, Corollary~\ref{cor:lowerboundnumneighbors}, and Lemma~\ref{lem:each-phase-not-too-many-neighbors}), the slowdown of the lazy random walk compared to the normal random walk is at most a constant factor.

        We now look at the effect of congestion on the time to complete random walks.
        By Lemma~\ref{lem:random-walk-time}, we see that when the degree is a constant, it takes $O(\log^3 n)$ rounds with high probability to finish running all the random walks walks in parallel and when the degree is $\Theta(\log n)$, it takes $O(\log^2 n)$ rounds with high probability to finish. In the first phase, we run these walks on a graph of constant degree and in subsequent phases, we run these walks on a graph of $\Theta(\log n)$ degree.
        In phase $\NUMPHASES$, each node performs $O(\log^2 n)$ lazy random walks, each of length $O(\log n)$, taking $O(\log^3 n)$ rounds with high probability. 
        Thus, the total run time of the algorithm is $\bigO{\log^3 n}$ rounds with high probability.
    \end{proof}
}

\section{Efficient Communication Protocols}
\label{sec:applications-graph}
In this section, we present efficient routing and broadcast algorithms for the graph $G^*$ that was constructed using the P2P protocol $\WEAVER$ in Section~\ref{sec:reconfig-alg}. 
Since the properties of $G^*$ hold with high probability,  the correctness of the protocols and the associated bounds
in the theorems hold with high probability.

\subsection{Efficient Broadcasting Protocols}
\label{subsec:reconfig-analysis-broadcast}
Let us assume that we are given a source node $source$ with a message that must be broadcast to every node in the graph. In this section, we design broadcast algorithms to be run on the graph $G^*$ that is constructed by the P2P construction protocol in Section~\ref{sec:reconfig-alg}. In order to argue about the efficiency of broadcast, we use \textit{broadcast propagation cost}, \textit{broadcast completion cost}, and \textit{broadcast completion time} (see Section~\ref{sec:defns}).

First,  we present a simple flooding-based broadcast algorithm called $\FLOODING$, in Section~\ref{subsubsec:simple-broadcast}, that has optimal broadcast propagation cost (up to polylog $n$ factors) and optimal broadcast completion cost (up to polylog $n$ factors) but at the expense of a very bad broadcast completion time. In particular, the broadcast propagation cost is $\Tilde{O}(\sqrt{n})$, the broadcast completion cost is $\tilde{O}(1)$, and the broadcast completion time is $\tilde{O}(\sqrt{n}$). From Theorem~\ref{the:lower-bound-broadcast}, we see that this broadcast propagation cost is asymptotically optimal up to polylog $n$ factors for any broadcast algorithm run by nodes uniformly distributed in Euclidean space. 

In order to obtain optimal bounds (up to polylog $n$ factors) for all three metrics, we design a  more sophisticated algorithm called $\BROADCAST$, in Section~\ref{subsubsec:broadcast-opt}, that requires that each node knows its own coordinates (instead of merely the distance between itself and some other node). $\BROADCAST$ has broadcast propagation cost $\tilde{O}(\sqrt{n})$, broadcast completion cost $O(1)$, and broadcast completion time $\tilde{O}(1)$. 

\subsubsection{Algorithm~\texorpdfstring{$\FLOODING$}{Geometric-Flooding}}\label{subsubsec:simple-broadcast}
~\\
\textbf{Brief Description.} The algorithm consists of each node participating in flooding over $G(\NUMPHASES)$. Initially, the source node sends the message to all its neighbors in $G(\NUMPHASES)$. Subsequently, each node, once it receives the message for the first time, transmits that message over each of its edges in $G(\NUMPHASES)$. 
~\\
\noindent \textbf{Analysis.} 
In $G(\NUMPHASES)$, each node has $O(\log n)$ neighbors and the weight of each edge is $O(r^\NUMPHASES)$. So, the sum of the edge weights in the graph, i.e., the broadcast propagation cost, is $O(n \cdot \log n \cdot r^\NUMPHASES) = O(n \cdot \log n \cdot \sqrt{\log n/n} ) = O(\sqrt{n \log^3 n})$.

The broadcast completion time corresponds to the diameter of the random geometric graph $G(\NUMPHASES)$. From~\cite{G18}, we see that the diameter of a random geometric graph $G(n,\rho)$ embedded in a unit grid is $\tilde{\Theta}(1/\rho)$. For the graph $G(\NUMPHASES)$, 
$\rho = \Theta(\sqrt{\log n/n})$. 
So the broadcast completion time is $\tilde{\Theta}(\sqrt{n})$. 

The broadcast completion cost is upper bounded by the product of diameter and edge weight, so it is $\tilde{O}(1)$.

The following theorem captures the relevant properties of the algorithm.

\begin{theorem}
    Algorithm~$\FLOODING$, when run by all nodes on $G^*$, results in a message being sent from a source node $source$ to all nodes in $\tilde{O}(\sqrt{n})$ broadcast completion time with broadcast completion cost $\tilde{O}(1)$ and broadcast propagation cost $O(\sqrt{n \log^3 n})$, which are all asymptotically optimal up to polylog $n$ factors.
\end{theorem}

\subsubsection{Algorithm~\texorpdfstring{$\BROADCAST$}{CompassCast}}\label{subsubsec:broadcast-opt}
Note that for this section, due to technical considerations, we assume that the parameter $r$ in the P2P construction protocol is chosen so that  $r \leq 0.25$ and $1/r$ is an integer. We additionally  assume that nodes know their own coordinates.

In order to describe the algorithm, we make use of the following notation for ease of explanation. Let $H_i$ represent the partition of the unit grid into a $1/r^i$ by $1/r^i$ grid of $1/r^{2i}$ equal size squares.

\longOnly{
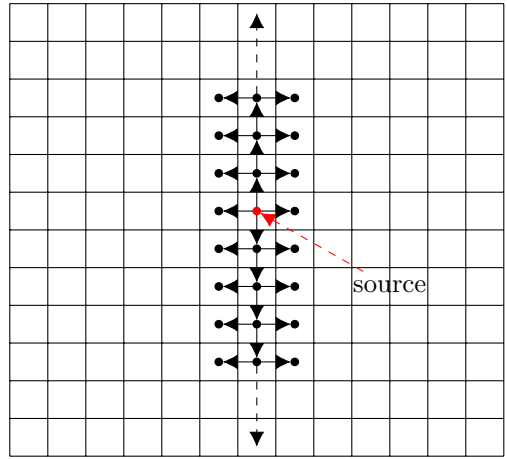
\begin{figure}
    \centering
    \begin{tikzpicture}[scale=1.0]
        \draw[step=0.5] (-1, -1) grid (5.5, 5);
        \node at (2.25, 5) (u5) {};
        \node[point] at (2.25, 3.75) (u4) {};
        \node[point] at (2.25, 3.25) (u3) {};
        \node[point] at (2.25, 2.75) (u2) {};
        \node[color=red,point] at (2.25, 2.25) (u1) {};
        \node[point] at (2.25, 1.75) (d1) {};
        \node[point] at (2.25, 1.25) (d2) {};
        \node[point] at (2.25, 0.75) (d3) {};
        \node[point] at (2.25, 0.25) (d4) {};
        \node at (2.25, -1) (d5) {};
        \draw[->] (u1) -- (u2);
        \draw[->] (u2) -- (u3);
        \draw[->] (u3) -- (u4);
        \draw[->, dashed] (u4) -- (u5);
        \draw[->] (u1) -- (d1);
        \draw[->] (d1) -- (d2);
        \draw[->] (d2) -- (d3);
        \draw[->] (d3) -- (d4);
        \draw[->, dashed] (d4) -- (d5);

        \node at (4, 1.25) (source) {source};
        \draw[dashed, red, ->] (source) -- (u1);

        \node[point] at (2.75, 3.75) (ur4) {};
        \node[point] at (2.75, 3.25) (ur3) {};
        \node[point] at (2.75, 2.75) (ur2) {};
        \node[point] at (2.75, 2.25) (ur1) {};
        \node[point] at (2.75, 1.75) (dr1) {};
        \node[point] at (2.75, 1.25) (dr2) {};
        \node[point] at (2.75, 0.75) (dr3) {};
        \node[point] at (2.75, 0.25) (dr4) {};
        \node[point] at (1.75, 3.75) (ul4) {};
        \node[point] at (1.75, 3.25) (ul3) {};
        \node[point] at (1.75, 2.75) (ul2) {};
        \node[point] at (1.75, 2.25) (ul1) {};
        \node[point] at (1.75, 1.75) (dl1) {};
        \node[point] at (1.75, 1.25) (dl2) {};
        \node[point] at (1.75, 0.75) (dl3) {};
        \node[point] at (1.75, 0.25) (dl4) {};

        \draw[->] (u1) -- (ur1);
        \draw[->] (u1) -- (ul1);
        \draw[->] (u2) -- (ur2);
        \draw[->] (u2) -- (ul2);
        \draw[->] (u3) -- (ur3);
        \draw[->] (u3) -- (ul3);
        \draw[->] (u4) -- (ur4);
        \draw[->] (u4) -- (ul4);
        \draw[->] (d1) -- (dr1);
        \draw[->] (d1) -- (dl1);
        \draw[->] (d2) -- (dr2);
        \draw[->] (d2) -- (dl2);
        \draw[->] (d3) -- (dr3);
        \draw[->] (d3) -- (dl3);
        \draw[->] (d4) -- (dr4);
        \draw[->] (d4) -- (dl4);
    \end{tikzpicture}
    \caption{\textbf{Phase one of $\BROADCAST$.} $H_2$ graph when $r=0.25$. In phase one, a message originating from $source$ is propagated to a node in the square North of it and a node in the square South of it. These nodes continue propagating the message North and South, respectively. Subsequently, each node that received a message and $S$ propagate the message to nodes in squares to the East and West of them. These nodes in turn continue propagating the message East and West, respectively.}
    \label{fig:broadcast-phase-one}
\end{figure}

\begin{figure}
    \begin{subfigure}[t]{0.4\textwidth}
    \begin{tikzpicture}[scale=1.0]


        \node[point,label=below:{$u$}] at (0.5, 0.5) (s) {};
        \node[point] at (0.8, 1.75) (10) {};
        \node[point] at (0.3, 2.45) (20) {};
        \node[point] at (0.6, 3.32) (30) {};

        \node[point] at (1.5, 0.5) (01) {};
        \node[point] at (1.8, 1.75) (11) {};
        \node[point] at (1.3, 2.45) (21) {};
        \node[point] at (1.6, 3.32) (31) {};

        \node[point] at (2.5, 0.5) (02) {};
        \node[point] at (2.8, 1.75) (12) {};
        \node[point] at (2.3, 2.45) (22) {};
        \node[point] at (2.6, 3.32) (32) {};

        \node[point] at (3.5, 0.5) (03) {};
        \node[point] at (3.8, 1.75) (13) {};
        \node[point] at (3.3, 2.45) (23) {};
        \node[point] at (3.6, 3.32) (33) {};
        \draw[step=1] (0, 0) grid (4, 4);

        \draw[->] (s) to[bend right] (10);
        \draw[->] (s) to[bend right] (20);
        \draw[->] (s) to[bend right] (30);
        \draw[->] (s) to[bend right] (01);
        \draw[->] (s) to[bend right] (11);
        \draw[->] (s) to[bend right] (21);
        \draw[->] (s) to[bend right] (31);
        \draw[->] (s) to[bend right] (02);
        \draw[->] (s) to[bend right] (12);
        \draw[->] (s) to[bend right] (22);
        \draw[->] (s) to[bend right] (32);
        \draw[->] (s) to[bend right] (03);
        \draw[->] (s) to[bend right] (13);
        \draw[->] (s) to[bend right] (23);
        \draw[->] (s) to[bend right] (33);
    \end{tikzpicture}
    \caption{\textbf{Phase two of $\BROADCAST$.} Cross section of one square of $H_i$ divided into smaller boxes corresponding to $H_{i+1}$ when $r=0.25$. In one step of phase two, node $u$ in a given square of $H_i$ identifies one node from each square of $H_{i+1}$ that intersects its square in $H_i$ and sends the message to it.} \label{fig:broadcast-phase-two}
    \end{subfigure}\hfill
    \begin{subfigure}[t]{0.4\textwidth}
    \begin{tikzpicture}[scale=1.0]
        \draw (0, 0) rectangle (4, 4);
        \node[point, label=below right:{$u$}] at (2, 2) (u) {};
        \node[point] at (3.6429502708368915, 1.8859988242670598) (0){};
        \node[point] at (1.4758298730955905, 0.0902896526206538) (1){};
        \node[point] at (0.35225768824200854, 3.4571967793850917) (2){};
        \node[point] at (0.9091502590336007, 1.1365111476572913) (3){};
        \node[point] at (1.9570096108112232, 2.854887212617063) (4){};
        \node[point] at (1.2262780445020405, 1.7890091396600067) (5){};
        \node[point] at (2.8820985883699883, 0.16481167801237984) (6){};
        \node[point] at (0.3277532282330893, 0.02005833312161931) (7){};
        \node[point] at (1.132350402090403, 2.278057460736687) (8){};
        \node[point] at (2.8720333095463086, 3.348034370701356) (9){};

        \draw[->] (u) -- (0);
        \draw[->] (u) -- (1);
        \draw[->] (u) -- (2);
        \draw[->] (u) -- (3);
        \draw[->] (u) -- (4);
        \draw[->] (u) -- (5);
        \draw[->] (u) -- (6);
        \draw[->] (u) -- (7);
        \draw[->] (u) -- (8);
        \draw[->] (u) -- (9);
    \end{tikzpicture}
    \caption{\textbf{Phase three of $\BROADCAST$.} Cross section of $B_u(r^{\NUMPHASES})$ centered on node $u$ shown. In phase three, each node $u$ with a message sends that message to all its neighbors in $B_u(r^{\NUMPHASES})$.}
    \label{fig:broadcast-phase-three}
    \end{subfigure}
    \caption{Phases two and three of $\BROADCAST$.}
\end{figure}
}

\noindent \textbf{Brief Description.}
The efficient broadcast of a message can be done in three phases.
In phase one, the message is propagated to exactly one node in each square of $H_2$ using $G(1)$.
Phase two is used to propagate the message to exactly one node in each square in $H_{\NUMPHASES}$ in a recursive manner as follows. Each node that received a message at the end of phase one takes ``ownership'' of all square of $H_3$ that lie within its square of $H_2$ and sends the message to exactly one node in each such square of $H_3$. In this manner, each node $u$ with the message in a square in $H_i$, $i \leq 2 < \NUMPHASES$, chooses one node per square of $H_{i+1}$ that lies within $u$'s square of $H_i$ and sends the message to them. Finally, exactly one node in each square of $H_{\NUMPHASES}$ will have the message. 
Phase three is used to propagate the message to every node in $G^*$ by having each node in the proceedings phase transmit the message to all its neighbors in $G^*$. Subsequently, each node that received the message further transmits it to all its neighbors in $G^*$.

\longOnly{%
\noindent \textbf{Detailed Description.}
We now describe in detail the three phases of algorithm~$\BROADCAST$.

In phase one, the message is propagated to exactly one node in each square of $H_2$ using $G(1)$. In this phase, when a node $u$ transmits the message to its neighbor $v$ in $G(1)$, in addition to the message, a subset of directions $N$, $S$, $E$, and $W$ denoting North, South, East, and West respectively are appended to the message. The source node $source$ lies in some square $s \in H_2$. Let $s_1 \in H_2$ ($s_2 \in H_2$) denote the square immediately North (South) of $s$ in $H_2$, assuming such a square exists. $source$ identifies one of its neighbors in $G(1)$ that lies in $s_1$ ($s_2$) and sends the message appended with the directions $N$, $W$, and $E$ ($S$, $W$, and $E$) to that node. Let $s_3 \in H_2$ ($s_4 \in H_2$) denote the square immediately West (East) of $s$ in $H_2$, assuming such squares exist. $source$ identifies one of its neighbors in $G(1)$ that lies in $s_3$ ($s_4$) and sends the message appended with the direction $W$ ($E$) to that node. Now, any node that receives the message, forwards it to its neighbors according to the following rules (note that for the following rules, if no square exists in the given direction, then the message is not forwarded):
\begin{enumerate}
    \item
          If the appended directions were $N$, $W$, and $E$, then the message is forwarded to a node in each square immediately to the North, West, and East of the given node's square with directions ($N$, $W$, and $E$), ($W$), and ($E$) respectively. 
    \item
          If the appended directions were $S$, $W$, and $E$, then the message is forwarded to a node in each square immediately to the South, West, and East of the given node's square with directions ($S$, $W$, and $E$), ($W$), and ($E$) respectively.
    \item
          If the appended direction was $W$, then the message is forwarded to a node in the square immediately to the West of the given node's square with direction $W$.
    \item
          If the appended direction was $E$, then the message is forwarded to a node in the square immediately to the East of the given node's square with direction $E$.
\end{enumerate}
See Figure~\ref{fig:broadcast-phase-one} for an example of phase one being run.

Phase two is used to propagate the message to exactly one node in each square in $H_{\NUMPHASES}$. Phase two consists of $\NUMPHASES-2$ stages.
In stage one, each node $u$ (that lies in square $y$) that received the message in phase one chooses one of $u$'s neighbors $v$  in $G(2)$ (that lies in square $y'$) for each square $y' \in H_3 \bigcap y$, and sends the message appended with the current stage number to $v$.\footnote{We use $H_3 \bigcap y$ to mean all those squares in $H_3$ that lie within the square $y$.}
More generally, in stage $i$, $1 \leq i \leq \NUMPHASES-2$, each node $u$ (that lies in square $y$) that received the message in stage $i-1$ chooses one of $u$'s neighbors in $G(i+1)$, say node $v$ (that lies in square $y'$), for each square $y' \in H_{i+2} \bigcap y$, and sends the message appended with the current stage number to $v$. 
See Figure~\ref{fig:broadcast-phase-two} for an example of phase two being run.

Phase three is used to propagate the message to every node in $G^*$. It consists of each node that received the message at the end of the last stage of phase two sending that message to each of its neighbors in $G(\NUMPHASES)$. Subsequently, each node that received the message subsequently transmits it to each of its neighbors. See Figure~\ref{fig:broadcast-phase-three} for an example of phase three being run.
}

\onlyShort{To save space, the analysis of algorithm~$\BROADCAST$, which yields Theorem~\ref{thm:broadcast-performance}, is left for the full version.}

\begin{theorem}\label{thm:broadcast-performance}
    Algorithm~$\BROADCAST$, when run by all nodes on $G^*$, results in a message being sent from a source node $source$ to all nodes with broadcast completion cost $O(1)$,  and $O(\log n)$ broadcast completion time  and broadcast propagation cost $O(\sqrt{n \log^3 n})$, which are asymptotically optimal up to polylog $n$ factors.
\end{theorem}

\onlyLong{
\noindent \textbf{Analysis.}
We need to show two things: (i) the algorithm is correct and (ii) the algorithm is efficient. First, let us focus on correctness. We must show that the given algorithm correctly achieves broadcast, i.e., at the end of the algorithm, all nodes receive the message. This is shown below.

Let us examine the algorithm phase by phase. Initially, only the source node $source$ has the message. For any square $y \in H_2$, define the column of $y$ as the set of all squares in $H_2$ that lie above and below $y$ within the same interval on the x-axis and define the row of $s$ as the set of all squares in $H_2$ that lie to the left and right of $y$ within the same interval on the y-axis.

$source$ lies in some square in $H_2$, say $s$. It is easy to see from the algorithm description that phase one causes the message to be passed recursively to each square in the column of $s$. Now, each square $s'$ that belongs to this column passes the message to each square in the row of $s'$. It is easy to see that no square receives the message twice.

To see that for each square in $H_2$, at least one node in that square receives the message, notice the following three observations. First, since $r/2 \geq 2r^2$, for each node $u$ that lies in some square $y \in H_2$, the area of the squares that lie to $y$'s North, South, East, and West are entirely encased within $B_u(r)$. Furthermore, each such square takes up $r^4/r^2 = r^2$ fraction of $B_u(r)$. From Observation~\ref{obs:nodes-nbrs-uniformly-distributed}, we see that since $\Omega(\log n)$ neighbors of $u$ are taken uniformly at random from $B_u(r)$, there is at least one neighbor of $u$ in $G(1)$ located in each of the four squares. Thus, we arrive at the following lemma.

\begin{lemma}\label{lem:broadcast-phase-one-correctness}
    Phase one of algorithm~$\BROADCAST$ results in exactly one node in each square of $H_2$ receiving the message.
\end{lemma}

Now, in phase two of the algorithm, each node in a square $y \in H_i$, $2 \leq i \leq \NUMPHASES-2$ forwards the message to a node in each square in $H_{i+1}$ that lies within $y$. Note that since $1/r$ is an integer, each square in $H_{i+1}$ lies completely within some square in $H_i$. Furthermore, since $r/2 \geq r^2$, it is easy to see that if node $u$ lies within a square $y \in H_i$, then all squares in $H_{i+1}$ that lie within $s$ are within $B_u(r^i)$. From Observation~\ref{obs:nodes-nbrs-uniformly-distributed}, we see that since $\Omega(\log n)$ neighbors of $u$ are taken uniformly at random from $B_u(r^i)$, for each of the $1/r^2$ squares of $H_{i+1}$ that lie within $y$, there is at least one neighbor of $u$ in $G(i)$ located in that square. Thus, we have the following lemma.

\begin{lemma}\label{lem:broadcast-phase-two-correctness}
    Phase two of algorithm~$\BROADCAST$ results in exactly one node in each square of $H_{\NUMPHASES}$ receiving the message.
\end{lemma}

Since $G(\NUMPHASES)$ is a random geometric graph such that any two nodes within distance $r^\NUMPHASES$ are neighbors, it is easy to see that once each square in $H_{\NUMPHASES}$ has a node that received the message, by having all such nodes transmit the message to their neighbors and neighbors of neighbors, all nodes in $G^*$ have received the message.

\begin{lemma}\label{lem:broadcast-phase-three-correctness}
    Phase three of algorithm~$\BROADCAST$ results in all nodes in $G^*$ receiving the message.
\end{lemma}

Now, we turn to the efficiency of the algorithm. Let us first analyze the number of hops needed by the message to reach any node, i.e., the broadcast completion time. The proof of the following lemma can be seen as the culmination of Lemmas~\ref{lem:broadcast-phase-one-correctness}, \ref{lem:broadcast-phase-two-correctness}, and \ref{lem:broadcast-phase-three-correctness}, and from simple geometric arguments.

\begin{lemma}\label{lem:broadcast-time-analysis}
    The algorithm ensures that all nodes receive the message in $O(\NUMPHASES)$ broadcast completion time and the broadcast completion cost is $O(1)$.
\end{lemma}

 \begin{proof}
     Phase one of the algorithm requires the message to be transmitted by a distance at most $1$ in the vertical direction and at most $1$ in the horizontal direction. When the message is passed vertically, for every two squares that the message moves through, a distance of at least $r^2$ is covered in the vertical direction. Thus, after $O(1/r^2) = O(1)$ hops, the message reaches one node per square in all squares belonging to the vertical that the source's square belongs to. A similar analysis shows that for each square $y \in H_2$, after $O(1)$ hops, the message reaches all squares in $y$'s row. Thus, after a total of $O(1)$ hops, the message is sent to one node in each square of $H_2$. Since each edge in $G(1)$ is used at most once in phase one, there is no congestion.

     Phase two consists of $\NUMPHASES - 2$ stages, where each stage consists of sending the message at most $1$ hop away. Once again, since each edge in $G(i+1)$ is used at most once in each stage, there is no congestion.

     Phase three consists of each node with the message sending that message at most two hops away. Again, since each edge in $G(\NUMPHASES)$ is used at most twice, there is no congestion.

     Thus, the message reaches all nodes in $O(\NUMPHASES)$ hops, i.e., the broadcast completion time is $O(\NUMPHASES)$.

     From the above analysis, we see that the total propagation cost along any given path from the source to any given destination is upper bounded in each phase as follows. In phase one, the message moves a maximum distance of $O(1)$ (maximum vertical distance and maximum horizontal distance). In phase two, the message moves through distance $O(\sum_{i=1}^{\NUMPHASES-2} r^{i+1}) = O(1)$. In phase three, the message moves through distance $O(r^{\NUMPHASES}) = O(1)$. Thus, the broadcast completion cost is $O(1)$.
 \end{proof}

Let us now analyze the broadcast propagation cost. We show that this algorithm, when run on the final obtained graph $G^*$, has a broadcast propagation cost that is asymptotically optimal up to polylog $n$ factors, i.e., the broadcast propagation cost of algorithm~$\BROADCAST$ on $G^*$ is $\Tilde{O}(\sqrt{n})$.

First we bound the broadcast propagation cost on $G(\NUMPHASES)$, as phase three of the algorithm involves a constant number of instances of broadcast on $G(\NUMPHASES)$.

\begin{lemma}\label{lem:broadcast-prop-cost-randomgeometricgraph}
    The broadcast propagation cost of algorithm  $\BROADCAST$ run on $G(\NUMPHASES)$ is $O(\sqrt{n \log^3 n})$.
\end{lemma}

\begin{proof}
    From Lemma~\ref{lem:min-max-bounds-nodes-space}, we see that each of the $n$ nodes in $G(\NUMPHASES)$ has $O(\log n)$ neighbors. Since the length of each edge in $G(\NUMPHASES)$ is at most $O(r^\NUMPHASES)$, we see that the sum of edge weights in the graph $G(\NUMPHASES)$ is   $O(n \cdot \log n \cdot r^\NUMPHASES) = O(n \cdot \log n \cdot \sqrt{\log n/n} ) = O(\sqrt{n \log^3 n})$.
\end{proof}

We now bound the broadcast propagation cost of algorithm~$\BROADCAST$ on $G^*$.

\begin{lemma}\label{lem:broadcast-prop-cost-finalgraph}
    The broadcast propagation cost of algorithm~$\BROADCAST$ on $G^*$ is $O(\sqrt{n \log^3 n})$.
\end{lemma}

\begin{proof}
    From Lemma~\ref{lem:broadcast-phase-one-correctness}, at the end of phase one, there are $O(1/r^4)$ nodes that contain the message. Each of these nodes passed the message to at most $4$ neighbors, and the length of each edge used is at most $O(r)$, since $G(1)$ is used. Thus, the sum of edge weights used in the first phase is $O(1/r^4 \cdot 4 \cdot r) = O(1/r^3)$.

    From Lemma~\ref{lem:broadcast-phase-two-correctness}, at the beginning of each stage $i$, $1 \leq i \leq \NUMPHASES - 2$, of phase two, there are $O(1/r^{(2i+2)})$ nodes that contain the message. Each of these nodes sends the message to $1/r^2$ neighbors, and the length of each edge used is at most $O(r^{(i+1)})$. Thus, the sum of edge weights in the second phase is $O(\sum_{i=1}^{\NUMPHASES -2} 1/r^{(2i+2)} \cdot 1/r^2 \cdot r^{(i+1)}) = O(1/r^{\NUMPHASES + 2}) = O(\sqrt{n/\log n})$.

    In phase three, we use $G(\NUMPHASES)$ to perform broadcast twice. From Lemma~\ref{lem:broadcast-prop-cost-randomgeometricgraph}, we see that sum of edge weights in the third phase while accounting for multiplicities of edges is $O(\sqrt{n \log^3 n})$.

    Thus, the broadcast propagation cost is $O(1/r^3 + \sqrt{n/\log n} + \sqrt{n \log^3 n}) = O(\sqrt{n \log^3 n})$.
\end{proof}
}

\subsection{An Efficient Routing Protocol}
\label{subsec:reconfig-analysis-greedy-routing}
In this section, we present an efficient routing algorithm, Algorithm~$\GREEDYROUTING$ (pseudocode in Algorithm~\ref{alg:greedy-routing}),  which allows us to
route a packet from any source $S$ to any destination $F$ in $\bigO{\log n}$ hops using $G^*$ such
that the path taken has propagation cost $\bigO{d(S,F)}$, where $d(S,F)$ is the Euclidean distance between $S$ and $F$.
An important property of this routing protocol is that it is localized and greedy: any node needs only local
information (of itself and its neighbors) to route a given message to its final destination. 
\longOnly{A brief description of the algorithm is given below. } 

\begin{algorithm}[h]
    \footnotesize
    \caption{Greedy Routing --- forwarding a message $m$ from node $S$ to node $F$.}
    \label{alg:greedy-routing}
    \begin{algorithmic}[1]
        \State $\texttt{current} \gets S$
        \State $\texttt{dist} \gets \infty$
        \While{$\texttt{current} \neq F$}
            \State Send a message to every \texttt{neighbor} of $\texttt{current}$ requesting $d(\texttt{neighbor}, F)$
            \ForEach{\texttt{neighbor} of \texttt{current}}
                \State $\texttt{new-dist} \gets d(\texttt{neighbor}, F)$
                \If{$\texttt{new-dist} < \texttt{dist}$}
                    \State $\texttt{dist} \gets \texttt{new-dist}$
                    \State $\texttt{closest-neighbor} \gets \texttt{neighbor}$
                \EndIf
            \EndForEach
            \State forward message to \texttt{closest-neighbor} (which then becomes $\texttt{current}$)
        \EndWhile
	\end{algorithmic}
\end{algorithm}

\shortOnly{Due to a lack of space, we give only the theorem statement below. A full analysis may be found in the full version.}

\longOnly{%
\noindent \textbf{Brief Description.} Let us assume that we are given a source node $S$ containing a message that must be sent to a destination node $F$. We append the ID of $F$ to the message and call the combined message $m$. Node $S$ looks at its neighbors in $G^*$ and chooses the neighbor that is closest, distance-wise, to $F$, call it node $u_1$, and sends $m$ to $u_1$. This process is repeated by $u_1$, resulting in $m$ being sent to another closer node, call it $u_2$. This process is repeated until $m$ is finally delivered to $F$.

\noindent \textbf{Analysis.}
First, we make a simple observation.
\begin{observation}\label{obs:each-node-enough-nbrs}
    For all $i$, $1 \leq i \leq \NUMPHASES$, each node $u$ chooses $\Omega( \log n)$  edges in $G_u(i)$ as a result of random walks.
\end{observation}

The following observation follows as each node $u$ chooses neighbors within $G(i)$ as a result of random walks that are run long enough to allow mixing within each $B_u(r^i)$.

\begin{observation}\label{obs:nodes-nbrs-uniformly-distributed}
    For all $i$, $1 \leq i \leq \NUMPHASES$, for each node $u$, the neighbors of $u$ that it chose in $G(i)$ are chosen uniformly at random within $B_u(r^i)$.
\end{observation}

The following lemmas show that there exists a node $u$ that can be reached from $S$ in $O(1)$ hops such that $F$ lies within $B_u(r)$.

\begin{lemma}\label{lem:distance-move-g1}
    If the message is currently at node $w$, and $w'$ is chosen as per Algorithm~\ref{alg:greedy-routing} such that $F$ does not lie within $B_{w'}(r)$, then $d(w',F) \leq d(w,F) - r/8$.
\end{lemma}

 \begin{proof}
     Let us assume that the message is currently at some node $w$ such that $F$ lies outside $B_w(r)$. Consider the line $\overline{wF}$ and let $m$ be the point on the line such that $d(w,m) = r/4$. Consider the line perpendicular to $\overline{wF}$ that passes through $m$. Let $w_1$ and $w_2$ be the points that lie on this perpendicular line at a distance of $r/8$ from $m$. Let $m'$ be the point  $\overline{wF}$ such that $d(w,m') = 3r/8$ and consider the points $w_1'$ and $w_2'$ such that $m'$ lies on $\overline{w_1'w_2'}$ and $w_1, w_2, w_1', w_2'$ form a rectangle. See Figure~\ref{fig:step-one} for an illustration.

     \begin{figure}[t]
         \centering
         \includegraphics[scale=0.4]{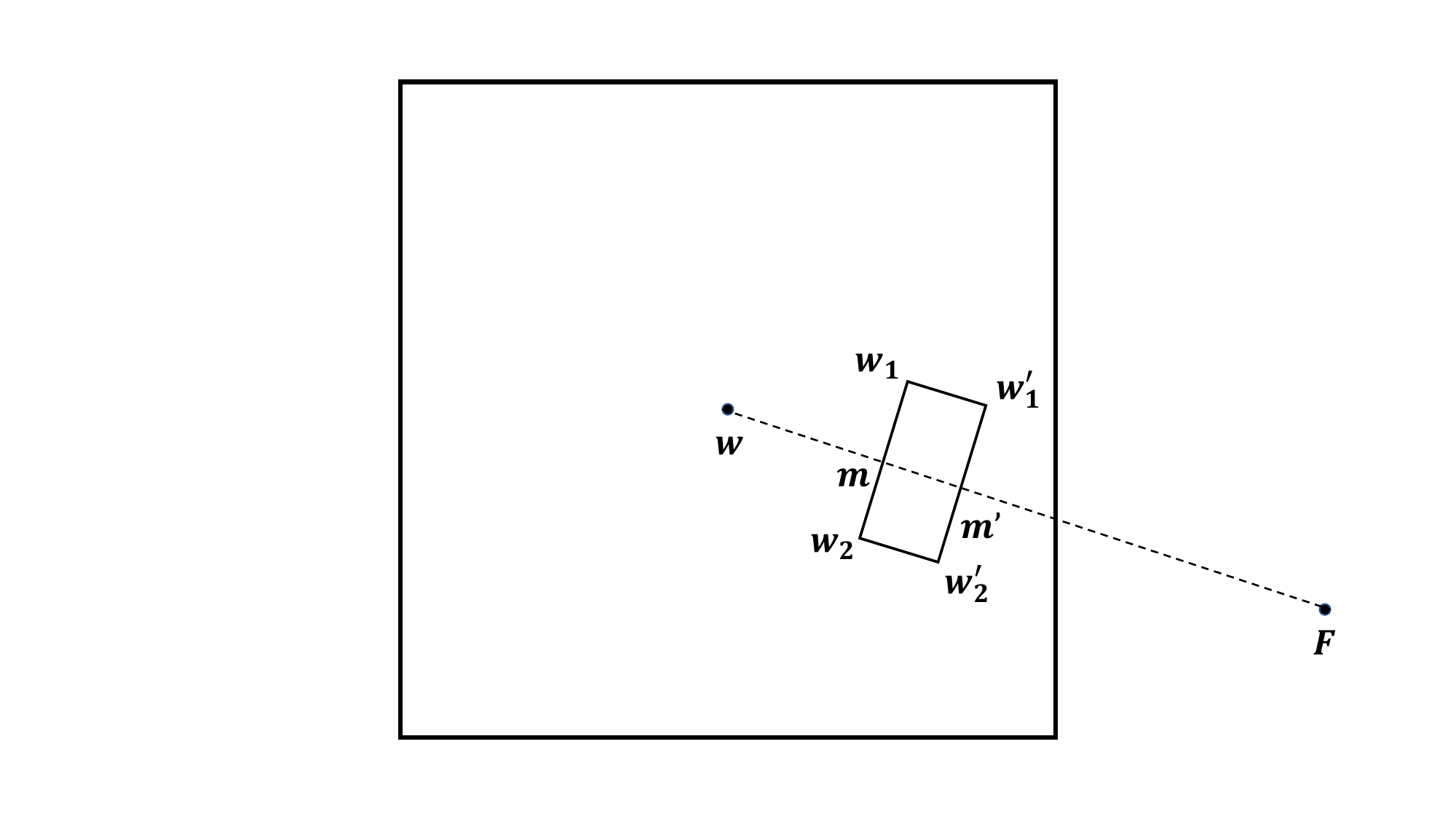}
         \caption{\footnotesize The region encompassed by the rectangle $w_1,w_2,w_1',w_2'$ is known as $A$. $d(w,m) = r/4$, $d(m,m') = r/8$, and $d(m,w_1) = d(m,w_2) = r/8$.}\label{fig:step-one}
     \end{figure}

     Notice that all points within the rectangle lie within a distance of $r/2$ from $w$ 
     and hence the entire region of the rectangle lies within $B_w(r)$.\footnote{Notice that the farthest points in the rectangle from $w$ are $w_1'$ and $w_2'$. $d(w,w_1') = d(w,w_2') = \sqrt{(\sfrac{r}{8})^2 + (\sfrac{3r}{8})^2} = \sfrac{\sqrt{10}}{8}r < \sfrac{r}{2}$.} Call this region $A$. Since the area of $A$ is a constant fraction of $B_w(r)$ ($A$ is at least $1/32$ of the area of $B_w(r)$) and $w$ has $\Omega(\log n)$ neighbors on expectation in $B_w(r)$, we see that with high probability, there is at least one neighbor of $w$, call it $w'$ in $A$.

     Now, we want to upper bound the distance $d(w',F)$, assuming $w'$ is some point in $A$. Recall that the algorithm has $w$ forward the message to its neighbor $v$ that is closest (by Euclidean distance) to $F$. As such, the distance $d(w',F)$ acts as an upper bound on the actual distance $d(v,F)$.

     Now, notice that within $A$, the worst possible location for $w'$ (resulting in the largest value of $d(w',F)$) is when $w'$ is either $w_1$ or $w_2$. In such a situation, 
         \begin{align*}d(w',F)
         &\leq d(w',m) + d(m,F) \hbox{ by the triangle inequality} \\
         &= \sfrac{r}{8} + d(w,F) - \sfrac{r}{4}  \hbox{ since $d(w,F) \geq \sfrac{r}{4}$}\\
         &\leq d(w,F) - \sfrac{r}{8}.\qedhere
     \end{align*}
 \end{proof}

Using Lemma~\ref{lem:distance-move-g1}, we prove the following lemma.

\begin{lemma}\label{lem:getting-near-v}
    Within $O(1)$ hops in $G(1)$, using Algorithm~\ref{alg:greedy-routing}, there exists a node $u$ that is reachable from $S$ such that $F$ lies within $B_u(r)$.
\end{lemma}

\begin{proof}
    From Lemma~\ref{lem:distance-move-g1}, we see that until the message reaches a node $u$ such that $F$ lies within $B_u(r)$, 
    the distance to $F$ (from the current node) decreases by at least $r/8$ every time the message is forwarded to a new node. More formally, if the message is currently at node $w$, and $w'$ is chosen as per \ref{alg:greedy-routing} such that $F$ does not lie within $B_{w'}(r)$, then $d(w',F) \leq d(w,F) - r/8$. Since $d(S,F) \leq \sqrt{2}$ and $r$ is a constant, we see that in $\lceil 8 \sqrt{2}/r \rceil = O(1)$ hops, we reach a node $u$ such that $F$ lies within $B_u(r)$.
\end{proof}

The above lemma shows that it takes $O(1)$ hops using $G(1)$ for a message to travel from $S$ to $u$. However, these are upper bounds on the number of hops taken by the greedy approach. If there exists a shorter path (e.g. via $G(0)$), both in hops and by extension propagation cost, from $S$ to some node $u'$ that satisfies the property that $F$ lies within $B_{u'}(r)$, that path will be taken instead.

The following two lemmas show that, once the message has reached node $u$ from Lemma~\ref{lem:getting-near-v}, the algorithm can find a path to $F$ using at most $O(1)$ edges from each of the graphs $G(i)$, $2 \leq i \leq \NUMPHASES$. The proof of the following lemmas is a straightforward generalization of the proofs of Lemma~\ref{lem:distance-move-g1} and Lemma~\ref{lem:getting-near-v}.

\begin{lemma}\label{lem:big-box-to-small-box-distance}
    If the message is currently at node $u$ such that $F$ lies within $B_u(r^i)$ but outside $B_u(r^{i+1})$, and the node that the message is forwarded to via Algorithm~\ref{alg:greedy-routing} is some $u'$, then $d(u',F) \leq d(u,F) - r^{i+1}/8$.
\end{lemma}

\begin{lemma}\label{lem:big-box-to-small-box-hops}
    For a node $u$ and some value of $i$, $1 \leq i < \NUMPHASES$, such
    that $F$ lies within $B_u(r^i)$ but outside $B_u(r^{i+1})$, within $O(1)$ hops in $G(i+1)$, using Algorithm~\ref{alg:greedy-routing}, there exists a node $u'$ that is reachable from $u$ such that $F$ lies within $B_{u'}(r^{i+1})$.
\end{lemma}

Recall that for every node $u$, all of the nodes that lie within $B_u(r^{\NUMPHASES})$ are neighbors of $u$ in $G(\NUMPHASES)$. Thus, once a node $v$ is reached such that $F$ lies within $B_u(r^{\NUMPHASES})$, the message can be directly forwarded to $F$ as $F$ is a neighbor of $v$ in $G(\NUMPHASES)$. Thus, we see that Algorithm~\ref{alg:greedy-routing} results in the message being routed from $S$ to $F$. Each subgraph is used to route the message for $O(1)$ hops and there are $\NUMPHASES$ such subgraphs. Thus, we have the following lemma.

\begin{lemma}\label{lem:greedy-routing-hops}
    Consider the graph $G^*$ obtained at the end of Algorithm~\ref{alg:reconfig}. For any
    source node $S$ and any destination node $F$, routing a packet from $S$ to
    $F$ using \ref{alg:greedy-routing} takes $\bigO{\NUMPHASES}$ hops.
\end{lemma}

Now we argue about the propagation cost of the path taken to route the message from $S$ to $F$.

\begin{lemma}\label{lem:greedy-routing-stretch}
    Consider the graph $G^*$ obtained at the end of Algorithm~\ref{alg:reconfig}. For any source node $S$ and any destination node $F$, the propagation cost of the routed path from $S$ to $F$ due to Algorithm~\ref{alg:greedy-routing} is $O(d(S,F))$, where $d(S,F)$ is the Euclidean distance between $S$ and $F$.
\end{lemma}

\begin{proof}
    It is either the case that (i) $F$ lies outside $B_S(r)$, (ii) $F$ lies within $B_S(r^i)$ but outside $B_S(r^{i+1})$ for some $1 \leq i < \NUMPHASES$, or (iii) $F$ lies in $B_S(r^{\NUMPHASES})$. We analyze each case separately and show that in each case, the propagation cost is $O(d(S,F))$.

    \noindent \textbf{Case (i) $F$ lies outside $B_S(r)$:} If $F$ lies outside $B_S(r)$, then $SF \geq r$. From Lemma~\ref{lem:getting-near-v} and Lemma~\ref{lem:big-box-to-small-box-hops}, we see that the message travels at most a constant number of hops in each $G(i)$, $1 \leq i \leq \NUMPHASES$. Let $c$ be the largest constant among all such constants. 
    Consider a straight line drawn from $S$ to $F$.
    Each of the at most $c$ hops in $G(i)$, $1 \leq i \leq \NUMPHASES$, say between some nodes $u$ and $v$ in $G(i)$ results in an at most $\sqrt{2} r^i$ additional cost to the propagation cost. Thus, the total propagation cost in a path taken by the message is $\leq c \sqrt{2}  \sum_{j=1}^{\NUMPHASES} r^j \leq c \sqrt{2}r/(1-r) = O(d(S,F))$.

    \noindent \textbf{Case (ii) $F$ lies within $B_S(r^i)$ but outside $B_S(r^{i+1})$ for some $1 \leq i < \NUMPHASES$:} If $F$ lies outside $B_S(r^{i+1})$ for some $1 \leq i < \NUMPHASES$, then $SF \geq r^{i+1}$. Using a similar analysis to Case (i), we see that the total propagation cost in a path taken by the message is $\leq c \sqrt{2}  \sum_{j=i+1}^{\NUMPHASES} r^j \leq c \sqrt{2}r^{i+1}/(1-r) = O(d(S,F))$.

    \noindent \textbf{Case (iii) $F$ lies in $B_S(r^{\NUMPHASES})$:} If $F$ lies within $B_S(r^{\NUMPHASES})$, then $F$ is an immediate neighbor of $S$ and as such, the propagation cost of the path is exactly $d(S,F)$.
\end{proof}

From Lemma~\ref{lem:greedy-routing-hops} and Lemma~\ref{lem:greedy-routing-stretch}, we arrive at our main theorem.
}

\begin{theorem}
    Consider the graph $G^*$ obtained at the end of Algorithm~\ref{alg:reconfig}. For any source node $S$ and any destination node $F$, routing a packet from $S$ to $F$ using Algorithm~\ref{alg:greedy-routing} takes $O(\log n)$ hops and the propagation cost of the routed path is $O(d(S,F))$, where $d(S,F)$ is the Euclidean distance between $S$ and $F$.
\end{theorem}
\section{Conclusion and Future Work}
\label{sec:conclusion}
We consider this work as a theoretical step towards the design and analysis of P2P topologies and associated communication
protocols. While our theoretical framework is only a rough approximation to real-world P2P networks, it provides a
rigorous model for the design and analysis of P2P protocols that takes into account propagation delays that depend on
not only the graph topology but also on the distribution of nodes
    across the Internet. Our model is inspired by several studies on the Internet, particularly
the Vivaldi system~\cite{vivaldi}, which posits how nodes on the Internet can be assigned coordinates in a
low-dimensional, even 2-dimensional, Euclidean space, that quite accurately captures the point-to-point latencies
between nodes. We have additionally performed empirical research, via simulation, on
the Bitcoin P2P network that suggests that the model is a reasonable approximation to a
real-world P2P network. We have also performed preliminary simulations of our routing and broadcast protocols
which broadly support our theoretical bounds. We leave a detailed experimental study for future work.

\newpage
\bibliographystyle{plainurl}
\bibliography{references}

\end{document}